%% file: BHT26.tex
\documentclass[authoryear,orivec,cleveref,thm-restate]{lipics-v2021}
\usepackage{graphicx} 
\usepackage{lineno}
\usepackage{url}
\usepackage{float}
\input{0macros}
\hideLIPIcs
\nolinenumbers
\authorrunning{L\'eonard Brice, Thomas A. Henzinger, K. S. Thejaswini}
\title{Algorithms for Equilibria in Concurrent Stopping Games}

\Copyright{L\'eonard Brice, Thomas A. Henzinger, K. S. Thejaswini}

\author{Léonard Brice}{Institute of Science \& Technology Austria \and \url{https://lnrdbrice.github.io/}}{leonard.brice@ist.ac.at}{https://orcid.org/0000-0001-7748-7716}{}

\author{Thomas A. Henzinger}{Institute of Science \& Technology Austria \and \url{https://pub.ista.ac.at/~tah/}}{tah@ist.ac.at}{https://orcid.org/0000-0002-2985-7724}{}

\author{K. S. Thejaswini}{Université Libre de Bruxelles \and \url{https://thejaswiniraghavan.github.io/}}{thejaswini.raghavan@ulb.be}{https://orcid.org/0000-0001-6077-7514}{}

\category{} 

\relatedversion{}

\EventEditors{C. Aiswarya, Ruta Mehta, and Subhajit Roy}
\EventNoEds{3}
\EventLongTitle{45th IARCS Annual Conference on Foundations of Software Technology and Theoretical Computer Science (FSTTCS 2025)}
\EventShortTitle{FSTTCS 2025}
\EventAcronym{FSTTCS}
\EventYear{2025}
\EventDate{December 17--19, 2025}
\EventLocation{BITS Pilani, K K Birla Goa Campus, India}
\EventLogo{}
\SeriesVolume{360}
\ArticleNo{34}
\begin{document}

\maketitle
\begin{abstract}

Concurrent games are a standard model for multi-agent systems,
with Nash equilibrium as their central solution concept. The associated
\emph{constrained existence problem}---does a game admit a Nash equilibrium
whose expected payoff lies within a prescribed interval for every
player?--- is undecidable, and remains so even for 10-player \emph{stopping} games, in which a
terminal state is reached almost surely under every strategy profile. We give
two routes to tractability.

We first relax exactness and consider the problem of approximate constrained existence problem, parametrised by $\varepsilon$-NE, which decides whether an
\(\varepsilon\)-Nash equilibrium with the prescribed payoffs exists.
The algorithm runs in exponential time, and only
polynomially in the bit-size of \(\varepsilon\). We complement it with a 
\PSPACE-hardness lower bound that holds already for turn-based games, and for
pure equilibria as well.

We then relax the solution concept, turning to \emph{extreme risk-sensitive
equilibria} (XRSE), recently introduced for turn-based stochastic games. Here the players
are partitioned into optimists and pessimists, who evaluate a strategy profile
by the best, respectively the worst, payoff attainable with positive
probability, instead of the expected payoff. We prove that the constrained existence problem for XRSE is
\NP-complete on concurrent games, as for turn-based games.

\end{abstract}
\section{Introduction}\label{sec:intro}
\input{1intro}

\section{Preliminaries}\label{sec:prelim}
\input{2prelims}

\section{An approximate algorithm for Nash equilibria} \label{sec:nash}
\input{3nash}

\section{An exact algorithm for XRSE} \label{sec:xe}
\input{4xe}

\section{Discussion} \label{sec:disc}
\input{5Disc}

\bibliography{biblio}

\appendix

\section{Appendix for \cref{sec:nash}}
\input{6Appendix3}

 \section{Appendix for \cref{sec:xe}}\label{app:xe}

\input{7Appendix4}

\end{document}

%% file: 0macros.tex
\usepackage{amsmath}
\usepackage{amssymb}
\usepackage{amsthm}
\usepackage{algorithm}
\usepackage{algorithmicx}
\usepackage[noend]{algpseudocode}
\usepackage[pdftex,dvipsnames]{xcolor}
\usepackage{tikzlings}
\usepackage{orcidlink}
\usepackage{tikzducks}
\usepackage{tikzsymbols}
\usepackage{tikz}
\usepackage{mathtools}
\usepackage{thmtools}
\usepackage{hyperref}
\usepackage{xspace}
\usepackage[new]{old-arrows}
\usepackage{float}
\usepackage{complexity}
\usepackage{caption}
\usepackage{subcaption}
\usepackage{graphicx}
\usepackage{mathabx}
\usepackage{stmaryrd}
\usepackage{wasysym}
\usepackage{pgfplots}
\pgfplotsset{compat=1.17}

\newcommand{\subp}[1]{\subparagraph*{#1.}}

\algdef{SE}[DOWHILE]{Do}{doWhile}{\algorithmicdo}[1]{\algorithmicwhile\ #1}%

\newcommand{\prob}{\mathbb{P}}

\newcommand{\reg}{\mathsf{reg}}
\newcommand{\dev}{\mathsf{dev}}

\newcommand{\Nb}{\mathbb{N}}
\newcommand{\Rb}{\mathbb{R}}
\newcommand{\Eb}{\mathbb{E}}
\newcommand{\Qb}{\mathbb{Q}}

\newcommand{\PM}{\mathbb{PM}}

\newcommand{\X}{\mathbb{X}}

\newclass{\TOWER}{TOWER}
\newclass{\ACKERMANN}{ACKERMANN}
\newclass{\EXPTIME}{EXPTIME}
\newclass{\IIEXPTIME}{2-EXPTIME}
\newclass{\NEXPTIME}{NEXPTIME}
\newclass{\SQRTSUM}{SqrtSum}
\newclass{\THREESAT}{3SAT}
\newclass{\PTIME}{PTIME}

\usepackage{graphicx}
\usepackage{array}
\usepackage{tikz}
\usetikzlibrary{automata, arrows, positioning, decorations.markings, decorations.pathreplacing}
\usetikzlibrary {decorations.pathmorphing, decorations.shapes}
\usetikzlibrary{backgrounds,fit,shapes.geometric}
\usetikzlibrary {graphs, quotes}
\usetikzlibrary{calc}
\usetikzlibrary{shapes,shapes.geometric,fit}
\usepackage[new]{old-arrows}
\usetikzlibrary {graphs,shapes.geometric,quotes}
\usetikzlibrary{decorations.pathreplacing}
\usetikzlibrary{automata}
\usetikzlibrary{calc}
\usetikzlibrary{arrows,automata,decorations.pathmorphing}
\usetikzlibrary{shapes,shapes.geometric,fit,calc,automata}
\usepackage{xargs} 
\usepackage{todonotes}
\newcommandx{\theju}[2][1=]{\todo[linecolor=blue,backgroundcolor=blue!25,bordercolor=blue,#1]{\tiny T: #2}}
\newcommandx{\leon}[2][1=]{\todo[linecolor=red,backgroundcolor=red!25,bordercolor=red,#1]{\tiny L:#2}}
\newcommandx{\aliA}[2][1=]{\todo[linecolor=green,backgroundcolor=red!25,bordercolor=red,#1]{\tiny L:#2}}

\usepackage{latexsym}
\usepackage{wasysym}

\newtheorem{problem}{Problem}

\theoremstyle{claimstyle}

\renewcommand{\Game}{\mathcal{G}}

\newcommand{\Cubes}{\mathcal{C}}

\newcommand{\Core}{\mathcal{C}}

\newcommand{\Eve}{\mathfrak{E}}

\newcommand{\Dist}{\mathcal{D}}

\newcommand{\bsigma}{{\bar{\sigma}}}
\newcommand{\btau}{{\bar{\tau}}}

\newcommand{\bchi}{{\bar{\chi}}}
\newcommand{\bx}{{\bar{x}}}
\newcommand{\by}{{\bar{y}}}
\newcommand{\bz}{\bar{z}}

\newcommand{\ba}{\bar{a}}

\newcommand{\val}{\mathsf{val}}

\newcommand{\Hist}{\mathsf{Hist}}

\newcommand{\punish}{\mathsf{punish}}
\newcommand{\Supp}{\mathsf{Supp}}
\newcommand{\anchor}{\mathsf{anchor}}
\newcommand{\last}{\mathsf{last}}

\newcommand{\Av}{\mathsf{Av}}
\newcommand{\irank}[2]{{#1}\mbox{-}\mathrm{rank}(#2)}
\DeclareMathOperator*{\argmin}{arg\,min}
\newcommand{\anc}[1]{\mathsf{anchor}_{#1}}

\newcommand{\Lab}{\Lambda}
\newcommand{\restr}{{\upharpoonright}}

\newcommand{\Chi}{\mathrm{X}}

\renewcommand{\epsilon}{\varepsilon}
\renewcommand{\phi}{\varphi}
\renewcommand{\l}{\ell}

\newcommand{\pay}[2]{\stackrel{\square}{#1}\stackrel{\circ}{#2}}
\newcommand{\act}[2]{\stackrel{\square}{#1}\stackrel{\circ}{#2}}

\renewcommand{\Circle}{\mathbin{\raisebox{-0.34ex}{\scalebox{1.5}{$\circ$}}}}

\tikzset{every node/.append style={initial text={}, node distance=15mm}}

\tikzset{stoch/.style={state, fill=black, inner sep=0.5pt, text=white, minimum size=0pt}}
\tikzset{circlev/.style={state,minimum size=20pt,fill, top color=white, bottom color=black!20}}
\tikzset{squarev/.style={state, rectangle, minimum size=20pt,fill, top color=white, bottom color=black!20}}
\tikzset{diamondv/.style={state, diamond, minimum size=20pt,fill, top color=white, bottom color=black!20}}
\tikzset{trianglev/.style={state, regular polygon, regular polygon sides=3, minimum size=15pt, inner sep=0.8pt, fill, top color=white, bottom color=black!20}}
\tikzset{pentagonv/.style={state, regular polygon, regular polygon sides=5, minimum size=20pt,fill, top color=white, bottom color=black!20}}
\tikzset{hexagonv/.style={state, regular polygon, regular polygon sides=6, minimum size=20pt,fill, top color=white, bottom color=black!20}}
\tikzset{gadget/.style={state, rectangle, rounded corners, dashed}}
\tikzset{mem/.style={state, minimum size=25pt, text=white, fill, top color=black!40, bottom color=black!60, rectangle, rounded corners}}

%% file: 1intro.tex
We now live surrounded by decentralised automated systems involving multiple agents---\emph{multi-agent systems}.
\emph{Stochastic games} constitute a powerful tool to model those, with applications in epidemic processes~\cite{Lef81}, formal verification~\cite{FKNP11}, learning theory~\cite{AJKS21}, cyber-physical systems~\cite{SEC16}, distributed and probabilistic programs~\cite{dAHJ01}, and probabilistic planning~\cite{TKI10}.

More expressive than their turn-based counterparts, \emph{concurrent} games are stochastic games meant to capture synchronous behaviors.
A play in a concurrent game can be described as a sequence of moves of a token on the vertices of a graph: from a given vertex, each player selects, simultaneously and independently, an \emph{action}.
The selected action profile induces a probability distribution over the outgoing edges: the token then moves to a new vertex according to that distribution.
In this paper, we consider \emph{terminal-reward} concurrent games (also known as \emph{simple quantitative} concurrent games): the graph contains some \emph{terminal vertices}, where the game stops, and each player receives a specified reward.
Preferences of each player among the different plays depend entirely on which terminal vertex is reached, no matter which path was taken.

In stochastic games, a natural solution concept is \emph{Nash equilibrium} (NE).
A strategy profile is a Nash equilibrium if no player can improve their expected payoff by deviating from their strategy, assuming the other players stick to theirs.
Nash equilibria are known to exist in stochastic games, but may be unsatisfying, as they yield poor payoffs to the players~\cite{CMJ04,Umm10}.
A more ambitious question is known as the \emph{constrained existence problem}: given a stochastic game and a constraint (usually given as bounds on the players' expected payoffs), does the game contain an NE that satisfies the constraint?
This problem hits the undecidability wall, even in turn-based games with a number of players fixed to 10~\cite{UW11}.
To overcome that barrier and find decidable variants, three natural approaches arise.

The first one consists in restricting either the class of games considered, or the class of strategies allowed.
Restrictions on the stochastic nature of the problem are unsatisfying: the problem remains undecidable when the players are not allowed to randomise, as well as when the game is deterministic~\cite{DBLP:conf/concur/UmmelsW11}.
It is only known to be decidable (in polynomial space) when those two restrictions apply simultaneously, thus ruling out any form of stochasticity~\cite{DBLP:journals/corr/BouyerBMU15}.
If the players are restricted to \emph{memoryless strategies}, i.e., play only according to the current vertex, then the problem reduces to deciding the validity of a formula in the existential theory of the reals, which can be done in polynomial space~\cite{HS20}. 
Such a restriction, although interesting for some formalisms, is of little help in settings such as control theory, where the game graph models an environment for a robot to traverse and the multiple players depict different independent controllers with their own objectives~\cite{BB6robotics,KB08,GNPW20}. Requiring memoryless strategies amounts to the absurdity of a robot repeating the same action every time it revisits a location; encoding the needed memory into the game graph itself causes a significant increase in the size of such an enhanced game graph~\cite{KB08}.

The second approach consists in designing \emph{approximate} algorithms, i.e., algorithms that return \emph{yes} on positive instances, and \emph{no} on instances that are sufficiently far from positive---but may return either outputs on instances that are close to the border (using the notion of $\epsilon$-NE), with a specified acceptable error term.
However, to this date, approximate algorithms have only been designed combined with the first approach: the approximate constrained existence problem of \emph{memoryless} NEs was recently shown to be in the class $\exists\Rb \cap \NP^\NP$ \cite{ABCT25}.

The third approach consists in looking for variants of the notion of NE itself, for which the equivalent problem would be decidable.
For classical refinements, such as subgame-perfect equilibria, an easy adaptation of the classical undecidability proof shows that the same result holds.
However, a new equilibrium notion was recently proposed for which an algorithm exists: \emph{extreme risk-sensitive equilibria}~\cite{BHT25} (XRSEs).
XRSEs are defined as Nash equilibria, with the slight difference that instead of maximising their \emph{expected} payoffs, the players intend to maximise either the minimal payoff they get with positive probability (\emph{pessimistic} players), or the maximal one (\emph{optimistic} players).
The constrained existence problem of XRSEs was shown decidable, and $\NP$-complete, in turn-based stochastic games, under randomised or pure strategies.
However, the case of concurrent games was left open.

\subparagraph*{Contributions.}
In this paper, we focus on a classical restriction of concurrent games, namely \emph{stopping} games.
A concurrent game is stopping if for every strategy profile, it is almost sure that some terminal vertex will be reached.
This notion rules out a specific case of plays, those that remain in the game forever, which is a sound restriction to model processes that are expected to eventually end, but with no duration bound.
The undecidability result cited above still applies with that restriction.
In this setting, our contribution is twofold.

First, we provide the first known approximate algorithm for the constrained existence problem of NEs, without any restriction on the players' strategies.
That algorithm relies on three ideas: (1) we show that every NE can be transformed into an $\epsilon$-NE \emph{with finite memory horizon}, meaning that it eventually follows a memoryless strategy profile.
Then, (2) we define the \emph{characteristic vector} of a strategy profile, containing the information of each player's expected payoff, and the maximal expected payoff they can obtain by deviating.
Thus, that vector is a finite-dimension object that contains the necessary data to decide whether the strategy profile witnesses a positive instance of the problem.
Finally, (3) we show how to iteratively construct an over-approximation of the set of characteristic vectors of strategy profiles with memory horizon $k$, for any integer $k$.
The over-approximation relies on a discretisation of the space of characteristic vectors: instead of computing the exact set, we compute a union of cubes that covers it.
Based on step (1), we compute a number $k$ (which depends on the error term) such that an over-approximation of strategy profiles with memory horizon $k$ is sufficient to decide the problem.
The resulting algorithm takes exponential time: thus, switching from the exact problem to the approximate version changes the complexity from undecidable to $\EXPTIME$.
We also give a $\PSPACE$ lower bound.

Our second contribution concerns XRSEs, whose constrained existence problem we prove is $\NP$-complete on concurrent stopping games (\cref{thm:Npcomplete}), matching the lower bound of the same problem in the turn-based setting~\cite[Lemma~20]{BHT25}; the work is in the $\NP$-upper bound.
It reduces, as in turn-based games, to showing that only polynomial memory for a witness XRSE to exist.
Concurrency however makes this step non-trivial and a na\"ive adaptation of the turn-based proof would lead to an exponential blow-up in the number of players. 
A further concurrent-specific obstacle is the punishment of players that deviate. When one player deviates, the other players can be assume to form a punishing coalition: however, this punishment may require randomisation, and the punishing players cannot resort to a common randomness source to do so, hence this setting is not reducible to a two-player one, unlike the assumptions in earlier work~\cite{AHK02,dAHK07}. Since the team in a concurrent game cannot correlate its randomisation, we appeal to recent results that show that memoryless team-punishment strategies without shared randomness~\cite{BHMST26}.





\subparagraph*{Structure of the paper.}
In \Cref{sec:prelim}, we introduce the necessary definitions.
In \Cref{sec:nash}, we give our approximate algorithm for Nash equilibria, and the lower bound.
In \Cref{sec:xe}, we give our exact algorithm for XRSEs.
We end with a discussion in \Cref{sec:disc}.

%% file: 2prelims.tex
\subp{Tuples}
For any set $X$ and any finite set $I$ (assumed clear from the context), a tuple $(x_i)_{i \in I}$ will usually be denoted by $\bx$; and conversely, when we are given a tuple $\bx$, we denote its $i$th coordinate as $x_i$.
For a fixed element $i \in I$, we also use the notation $\bx_{-i} = (x_j)_{j \in I \setminus \{i\}}$, and denote by $(\bx_{-i}, x'_i)$ the tuple whose $i$th element is $x'_i$, and whose $j$th element, for every $j \neq i$, is $x_j$.
When the set $I$ corresponds to players, we use the word \emph{profile} instead of \emph{tuple}.

\subp{Distributions and Support} 
For a finite set $X$, we denote by $\Dist(X)$ the set of probability distributions over $X$, i.e. the set of mappings $d: X \rightarrow [0,1]$ satisfying $\sum_{x \in X} d(x) = 1$.
For a given distribution $d$, we denote the \emph{support} of $d$ as $\Supp(d) = \{x \in X \mid d(x) > 0\}$.

\subp{Games} 
In this paper, we focus on \emph{terminal-reward concurrent stochastic multiplayer games}.
When the context is clear, we simply call them \emph{concurrent games} or \emph{games}.

\begin{definition}[Concurrent game]
    A \emph{concurrent game} is a tuple:
    $$\Game = (V, v_0, T, \Pi, (A_i)_{i \in \Pi}, (\Av_i)_{i \in \Pi}, \Delta, \mu)$$
    consisting of:
\begin{itemize}
    \item a finite set $V$ of \emph{vertices}, containing an \emph{initial vertex} $v_0$ and a subset $T \subseteq V$ of \emph{terminal vertices};
    
    \item a finite set $\Pi$ of \emph{players};
    
    \item for each player $i \in \Pi$, a finite set $A_i$ of \emph{actions}, and a mapping $\Av_i: V \setminus T \to A_i$ that maps each vertex $v$ to the set of actions \emph{available} to player $i$ in $v$.
    We also write $A = \prod_i A_i$ for the set of action profiles, and $\Av(v) = \prod_{i \in \Pi} \Av_i(v)$ for the action profiles available in $v$;
    
    \item a \emph{transition function} $\Delta: (V \setminus T) \times \prod_{i \in \Pi} A_i \rightarrow \mathcal{D}(V)$.
    For an action profile $\ba$ played at vertex $v$, the quantity $\Delta(v, \ba)(w)$ denotes the probability of transitioning to vertex $w$;
    
    \item a payoff function $\mu: T \rightarrow \mathbb{R}^\Pi$ assigning a reward profile to each terminal vertex. For each player $i$ and terminal vertex $t$, we denote by $\mu_i(t)$ the $i$th coordinate of the profile $\mu(t)$.
\end{itemize}
\end{definition}

As specific case, a game is \emph{deterministic} if for every vertex $v$ and action profile $\ba$, the support of the distribution $\Delta(v, \ba)$ is a singleton.
It is \emph{turn-based} if from every $v$, only one player faces a non-trivial choice, i.e., the set $\Av_i(v)$ is a singleton for all players except one.
A \emph{Markov decision process} (MDP) is a game with one player.
A \emph{Markov chain} is a game with zero players.

\subp{A note on representation} 
For each vertex, the set of transition probabilities is given as a matrix with dimension $|\Pi|$.
The size of the transition matrix is exponential in the number of players that have more than one distinct action available at that vertex. 
If every player has more than one action available, this leads to 
a size exponential in the number of players; but if at most $c$ players have a non-trivial choice in each state, where $c$ is a constant, then the space needed to represent the transition function is bounded by $O(|V| (\max_i |A_i|)^c)$.

\subp{Plays and histories} 

A \emph{play} in the game $\Game$ is a sequence of states and action profiles $\pi = (v_0, \ba_0, v_1, \dots)$ that is either infinite, or ends with a terminal vertex, i.e. a word in the set $v_0 A (VA)^\omega \cup v_0 A (VA)^*T$.
A \emph{history} is a finite prefix $h$ of a play that ends with a vertex, i.e. a word in the set $(VA)^*V$.
Given a history $h$, we denote the last vertex in $h$ by $\last(h)$.
The set of histories is denoted by $\Hist$.

The payoff function $\mu$ is extended to plays, by defining $\mu(\pi) = \mu(t)$ if the play $\pi$ eventually reaches the terminal vertex $t$, and $\mu(\pi) = (0)_{i \in \Pi}$ if $\pi$ is infinite.

\subp{Strategies}

A \emph{strategy} for player $i$ is a mapping $\sigma_i$ that takes a history $h$ ending in a non-terminal state $v$, and outputs a probability distribution over the set $\Av_i(v)$.
A strategy is \emph{pure} if for every history $h$, there is an action $a$ with $\sigma_i(h)(a) = 1$.
It is \emph{memoryless} if for every two histories $h, h'$, we have $\sigma_i(h) = \sigma_i(h')$ whenever $\last(h) = \last(h')$.
It is \emph{positional} if it is simultaneously memoryless and pure.
These definitions extend to strategy profiles, and given a history $h$ and a strategy profile, we abusively write $\bsigma(h)$ for the distribution $\ba \mapsto \prod_i \sigma_i(a_i)$.
Given a strategy $\sigma_i$ and a history $hv$, we denote by $\sigma_{i \restr hv}$ the residual strategy the maps each history $h'$ starting in $v$ to the distribution $\sigma_i(hh')$.

A play $\pi = (v_0, \ba_0, v_1, \dots)$ is \emph{compatible} with a strategy $\sigma_i$ if for every step $k$, there exists an action profile $\ba \in \prod_{i \in \Pi} \Supp(\sigma_i(\pi_{\le k}))$ such that $v_{k+1} \in \Supp(\Delta(v_k, \ba))$.
This definition extends naturally to strategy profiles, and to histories.

A complete strategy profile $\bsigma = (\sigma_i)_{i \in \Pi}$ defines a probability distribution $\prob_\bsigma$ over plays in the game $\Game$.
Then, we liberally identify every history $h$ to the event:
$$\{\pi \mid \text{the history } h \text{ is a prefix of } \pi\},$$
by writing for instance $\prob_\bsigma(h)$ for the probability that the history $h$ is followed.

Under such a probability measure, the payoff function $\mu$, extended to plays, can be seen as a (measurable) random variable ranging over the space $\Rb^\Pi$.
We write $\Eb(\bsigma)$ for its expected value, and $\Eb_i(\bsigma)$ for player $i$'s expected payoff.

\subp{Stopping Games} 
In this work, we restrict our attention to \emph{stopping games}.
A game is \emph{stopping} if for every strategy profile $\bsigma$, the probability of reaching the set $T$ under $\bsigma$ is exactly 1.
This hypothesis guarantees that the set of infinite plays has a probability measure of 0, independently from the strategy profile considered.
Consequently, the assumption $\mu(\pi) = (0)_{i \in \Pi}$ for infinite plays $\pi$ has no impact on expected payoffs.

\subp{Nash equilibria and their constrained existence problem}
The most classical solution concept in multiplayer games is \emph{Nash equilibrium}.

\begin{definition}[Nash equilibrium]
    A strategy profile $\bsigma$ in the game $\Game$ is a \emph{Nash equilibrium} if and only if for each player $i$ and every strategy $\sigma'_i$, we have $\Eb_i(\bsigma_{-i}, \sigma'_i) \leq \Eb_i(\bsigma)$.
\end{definition}

Then, a natural decision problem to study is the following one.

\begin{problem}[Constrained existence problem of NEs]
    Given a concurrent game $\Game$, and two payoff profiles $\bx, \by \in \Rb^\Pi$, does there exist a Nash equilibrium $\bsigma$ in $\Game$ satisfying $x_i \leq \Eb_i(\bsigma) \leq y_i$ for every player $i$?
\end{problem}

However, it is known from the work of Ummels and Wojczak~\cite[Theorem~4.9]{UW11} that this problem is undecidable.
Their proof goes by a reduction from the halting problem of two-counter machines.
A careful reading of that proof shows that the game that is constructed is always stopping, hence this result holds in our more restricted setting.

\begin{theorem}[\cite{UW11}]
    The constrained existence problem of NEs is undecidable, even in stopping games.
\end{theorem}

Here, we overcome this difficulty by exhibiting two new decidable variants: the approximated version of this problem, and the constrained existence problem of XRSEs.



%% file: 3nash.tex
\subsection{The problem}

As discussed in the introduction, the constrained existence problem of Nash equilibria is undecidable.
In this section, we therefore consider its approximate version, and present the first known algorithm to decide it.

\begin{problem}[Approximate constrained existence problem of NEs in stopping games] \label{pb:ne}
    Given a stopping game $\Game$, two payoff profiles $\bx, \by \in \Qb^\Pi$, and a rational quantity $\epsilon > 0$, such that we have the guarantee that either:
    \begin{itemize}
        \item there exists a Nash equilibrium $\bsigma$ in $\Game$ satisfying $x_i \leq \Eb_i(\bsigma) \leq y_i$ for each player $i$,
        \item or there exists no $\epsilon$-Nash equilibrium satisfying $x_i - \epsilon \leq \Eb_i(\bsigma) \leq y_i + \epsilon$ for each player $i$,
    \end{itemize}
    are we in the first case?
\end{problem}

We assume an instance of the problem, and denote the number of non-terminal vertices in the game $\Game$ by $N$, the maximum reward absolute value by $R = \max_{i\in\Pi} \max_{t \in T} |\mu_i(t)|$, and the least transition probability $p_{\min} = \min_{v\in V} \min_{\ba \in \Av(v)} \min_{w \in \Supp(\Delta(v, \ba))} \Delta(v, \ba)(w)$.
The quantity $\epsilon$ is sometimes called \emph{error term}.

\subsection{A bound on stopping probabilities}

The game $\Game$ is assumed stopping: under every strategy profile, it is almost sure that some terminal vertex will eventually be reached.
The following lemma shows that the probability of avoiding terminal vertices for a given amount of time can be bounded.

\begin{lemma}
    Let $\lambda = 1- p_{\min}^N$.
    For every strategy profile $\bsigma$ and every $k \in \Nb$, we have $\prob_\bsigma(V^{kN+1}) \leq \lambda^k$.
\end{lemma}

\begin{proof}
    Let us consider the MDP obtained from $\Game$, by merging all players into one agent selecting all players' actions.
    Consider the objective of avoiding terminal vertices for at least $kN+1$ steps.\footnote{By \emph{steps}, we always mean the number of vertex visits---the number of action interactions is then $kN$.}
    By standard probabilistic theorems~\cite[Theorem 4.4.2]{DBLP:books/wi/Puterman94}, the agent has a pure strategy $\sigma^0$ that maximises the probability of satisfying this objective.
    Now, from any vertex $v$, when following that strategy, since the game is stopping, it is almost sure that a terminal will be reached.
    Therefore, there is a path from $v$ to $T$ that is compatible with $\sigma^0$; and that path can be chosen of length at most $N+1$.
    Since the strategy $\sigma^0$ is pure, each transition along that path is followed with probability at least $p_{\min}$, which means that said path is followed with probability at least $p_{\min}^N$.
    Consequently, from any path $v$, when following $\sigma^0$, the probability of avoiding $T$ for $N+1$ steps is smaller than or equal to $1-p_{\min}^N = \lambda$.
    By iterating this reasoning, the probability of avoiding $T$ for $kN+1$ steps is smaller than or equal to $\lambda^k$.
    By optimality of $\sigma^0$, the same bound applies to any strategy profile from $v_0$.
\end{proof}

    \subsection{Strategy profiles with finite memory horizon}

In this subsection, we show that Nash equilibria can be approximated by $\epsilon$-NEs with finite \emph{memory horizon}, a special case of finite-memory strategy profiles.

\begin{definition}[Memory horizon]
    A strategy profile \emph{with memory horizon $k$} is a strategy profile $\bsigma$ such that for every history $h$ of length $k$, there exist memoryless strategy profile $\btau^h$, such that for every history $h'$ of which $h$ is a prefix, we have $\bsigma(h') = \btau(h')$.
\end{definition}

As we will see, a succinct representation of strategy profiles with finite horizon can be constructed in a finite amount of time.
This yields us an algorithm thanks to the following lemma: an NE can be approximated by an $\epsilon$-NE with finite memory horizon, by following the NE for some number of steps and then truncate it, to follow memoryless strategies.

\begin{restatable}[App.~\ref{app:truncation}]{lemma}{lmTruncation}\label{lm:truncation}
    Let $\bsigma$ be a Nash equilibrium in the game $\Game$.
    Let $\delta > 0$, and let $k$ such that $\lambda^k \leq \frac{\delta}{4R}$.
    Then, there exists a $\delta$-Nash equilibrium with memory horizon $kN$ where each player has an expected payoff that is at most different by $\delta/2$.
\end{restatable}

In all what follows, we fix $K$ as the least integer satisfying $\lambda^K \leq \frac{\epsilon}{8R}$---such that for every NE, there a $\frac{\epsilon}{2}$-NE with memory horizon $KN$ that has the same expected payoffs.
That number grows exponentially with the instance size.

\begin{lemma}
    The number $K$ grows exponentially with the game size, and polynomially with the bit-size of $\epsilon$.
\end{lemma}

\begin{proof}
    We have:
    $$K = \left\lceil \frac{-\log\left(\frac{\epsilon}{8R}\right)}{-\log(\lambda)} \right\rceil \leq \frac{\log\left(\frac{\epsilon}{8R}\right)}{\log(\lambda)} + 1.$$
    On the one hand, the quantity $-\log\left(\frac{\epsilon}{8R}\right)$ is linear in the bit sizes of $\epsilon$ and $R$.
    On the other hand, we have
    $-\log(\lambda) = -\log(1-p_{\min}^N) \sim p_{\min}^N = O\left(2^{\|\Game\|}\right).$
\end{proof}

    \subsection{Characteristic vectors}

The main idea of our algorithm is that although a history-dependent strategy profile may not be described with finite space, only a small amount of data is needed to decide whether it witnesses a positive instance.
That data is captured by its \emph{characteristic vector}.

\begin{definition}[Characteristic vector]
    Given a game $\Game$ and a strategy profile $\bsigma$ in $\Game$, we define the \emph{characteristic vector} of $\bsigma$ as the vector $\bchi^\bsigma \in \Rb^{\Pi \times \{\reg, \dev\}}$ where for each player $i$, we have $\chi_{i\reg} = \Eb_i(\bsigma)$, and $\chi_{i\dev} = \sup_{\sigma'_i} \Eb_i(\bsigma_{-i}, \sigma'_i)$.
\end{definition}

Note that the strategy profile $\bsigma$ is an $\epsilon$-NE if and only if we have $\chi^\bsigma_{i\dev} \leq \chi^\bsigma_{i\reg} + \epsilon$.

    \subsection{An over-approximation of finite-memory-horizon strategy profiles}

We now show how to construct, for every $k$, an over-approximation of the set of characteristic vectors of strategy profiles with memory horizon $k$.
We fix $D = 2(K+1)\cdot\frac{R}{\epsilon}$, and discretise the space $[-R, R]^{\Pi \times \{\reg, \dev\}}$ into $D^{2|\Pi|}$ cubes of the form:
$$C = \prod_{\gamma \in \Pi \times \{\reg, \dev\}} \left[R\frac{d_\gamma}{D}, R\frac{d_\gamma+1}{D}\right],$$
where each $d_\gamma$ is an integer ranging from $-D$ to $D-1$.
The set of those cubes is denoted $\Cubes$.

\begin{definition}[The sequence $(\Chi_k)_k$] \label{def:Chin}
    We define the sequence $(\Chi_k)_k$, where each $\Chi_k$ is a mapping that maps each vertex to a subset of $[-R, R]^{\Pi \times \{\reg, \dev\}}$, inductively as follows.
    First, for each vertex $v$, the set $\Chi_0(v)$ is the union of all cubes $C \in \Cubes$ that contain a vector $\bchi^\bsigma$, where $\bsigma$ is a memoryless strategy from $v$.
    Then, for each $k$, the set $\Chi_k(v)$ is the union of all cubes $C \in \Cubes$ such that there exist:
    \begin{itemize}
        \item a vector $\bchi \in C$;
    
        \item vectors $\bchi^{\ba w} \in \Chi_{k-1}(w)$ for each action profile $\ba \in \Av(v)$ and $w \in V$,
        
        \item and coefficients $\alpha_{ia} \in [0,1]$ for each player $i$ and action $a \in \Av_i(v)$,
    \end{itemize}
    such that, for each player $i$:
    \begin{itemize}
        \item we have $\sum_{a  \in \Av_i(v)} \alpha_{ia} = 1$;

        \item we have $\chi_{i\reg} = \sum_{\ba \in \Av(v)} \sum_{w \in V} \Delta(v, \ba)(w) \left(\prod_j \alpha_{ja_j} \right) \chi^{\ba w}_{i\reg}$;

        \item and we have $\chi_{i\dev} = \max_{a_i \in \Av_i(v)} \sum_{\ba_{-i} \in \Av_{-i}(v)} \sum_{w \in V} \Delta(v, \ba)(w) \left(\prod_{j \neq i} \alpha_{ja_j} \right)  \chi^{\ba w}_{i\dev}$.
    \end{itemize}
\end{definition}

Intuitively, each coefficient $\alpha_{ia}$ represents the probability that player $i$ performs action $a$ at the first step.
Then, the vector $\bchi^{\ba w}$ represents the characteristic vector of the strategy profile that is followed after performing the action profile $\ba$ and reaching the vertex $w$: the vector $\bchi$ is then the characteristic vector of a strategy profile from $v$.
Since the whole cube $C$ is added to the set $\Chi_k(v)$, and not only the vector $\bchi$, this set is an over-approximation of the set of characteristic vectors of strategy profiles with memory horizon $k$; but since the cube size was chosen small enough, the accumulated errors is controlled.

\begin{restatable}[App.~\ref{app:Chin_meaning}]{lemma}{lmChinMeaning}\label{lm:Chin_meaning}
    For every $k$ and each $v$, the set $\Chi_k(v)$ contains all vectors of the form $\bchi^\bsigma$, where $\bsigma$ is a strategy profile from $v$ with memory horizon $k$.
    Conversely, for every vector $\bchi \in \Chi_k(v)$, there exists a strategy profile $\bsigma$ from $v$ with memory horizon $k$ satisfying $\| \bchi - \bchi^\bsigma \| \leq (k+1)\frac{R}{D}$.
\end{restatable}

Finally, let us show that those sets can be computed in exponential time.

\begin{restatable}[App.~\ref{app:computing_Chin}]{lemma}{lmComputingChin}\label{lm:computing_Chin}
    Given the game $\Game$, a vertex $v$ and a number $k$, the set $\Chi_k(v)$ can be computed in exponential time.
\end{restatable}



\subsection{Algorithm}

We can now use the results proven above to present an algorithm.

\begin{theorem}
    The approximate constrained existence problem of NEs in stopping games is decidable in time exponential in the game size, and polynomial in the bit-size of $\epsilon$.
\end{theorem}

\begin{proof}
    We first present the algorithm, then prove its correctness, and later on its complexity.

    \subparagraph*{The algorithm.}
    For $k$ ranging from $0$ to $KN$, compute all sets $\Chi_k(v)$.
    If, for some $k$, there is a vector $\bchi \in \Chi_k(v_0)$ such that $x_i - \frac{\epsilon}{2} \leq \chi_{i\reg} \leq y_i + \frac{\epsilon}{2}$ and $\chi_{i\dev} \leq \chi_{i\reg} + \frac{\epsilon}{2}$ for each player $i$, then stop the algorithm and return \emph{yes}.
    Otherwise, return \emph{no}.

    \subparagraph*{If the algorithm returns \emph{yes}, then we have a positive instance.}
    Assume it returns \emph{yes} at step $k$.
    Let $\bchi \in \Chi_k(v_0)$ be the vector detected.
    By \Cref{lm:Chin_meaning}, there is a strategy profile $\bsigma$ from $v_0$ such that $\left\| \bchi - \bchi^\bsigma \right\| \leq (k+1)\frac{R}{D} \leq (K+1)\frac{R}{D} = \frac{\epsilon}{2}$, by definition of $D$.
    Then, for each player $i$, we have:
    $$\chi^\bsigma_{i\reg} \leq \chi_{i\reg} + \frac{\epsilon}{2} \leq y_i + \frac{\epsilon}{2} + \frac{\epsilon}{2} \leq y_i + \epsilon$$
    by the stopping criterion; and similarly, we have $\chi^\bsigma_{i \reg} \geq x_i - \epsilon$, and:
    $$\chi^\bsigma_{i\dev} \leq \chi_{i\dev} + \frac{\epsilon}{2} \leq \chi_{i\reg} + \frac{\epsilon}{2} + \frac{\epsilon}{2} = \chi_{i\reg} + \epsilon.$$
    This shows that the strategy profile $\bsigma$ is an $\epsilon$-NE with $x_i-\epsilon \leq \Eb_i(\bsigma) \leq y_i+\epsilon$ for each player $i$, and by the guarantee given on the problem instances, that we have a positive instance.

    \subparagraph*{If the algorithm returns \emph{no}, then we have a negative instance.}

    If the algorithm returns \emph{no}, then the set $\Chi_{KN}(v_0)$ contains no vector satisfying the stopping criterion.
    By \Cref{lm:Chin_meaning} again, this means there is no $\frac{\epsilon}{2}$-NE with memory horizon $KN$ satisfying $x_i-\frac{\epsilon}{2} \leq \Eb_i(\bsigma) \leq y_i+\frac{\epsilon}{2}$ for each player $i$.
    Now, by \Cref{lm:truncation}, 
    applied with $\delta = \frac{\epsilon}{2}$, any Nash equilibrium satisfying $x_i \le \Eb_i(\bar{\sigma}) \le y_i$ for each player $i$ would yield an $\frac{\epsilon}{2}$-NE with memory horizon $KN$ whose expected payoffs lie in $[x_i - \frac{\varepsilon}{4},\, y_i + \frac{\varepsilon}{4}]
\subseteq
[x_i - \frac{\varepsilon}{2},\, y_i + \frac{\varepsilon}{2}]$: there
is therefore no such Nash equilibrium, that is, we have a negative
instance.

    \subparagraph*{Complexity.}
    The maximal number of steps of this algorithm is $K$, which is exponential in the size of the game and linear in the bit-size of $\epsilon$.

    At each step, by \Cref{lm:computing_Chin}, the computation of the sets $\Chi_k(v)$ takes exponential time in the size of $\Game$.
    Checking the stopping criterion amounts to checking, for each cube $C \subseteq \Chi_k(v)$ whether it intersects the polytope of vectors satisfying the stopping criterion: the can be done in a time polynomial in the size of the systems of inequalities representing those polytopes, i.e. in time polynomial in the instance size.
    Since the set $\Chi_k$ may contain exponentially (in the size of $\Game$) many cubes, checking the stopping criterion takes time exponential in the size of $\Game$ and polynomial in the size of $\epsilon$.
    Overall, this algorithm shows the desired complexity.
\end{proof}

\subsection{The pure case}

We quickly turn our attention to a variant of \Cref{pb:ne}, the constrained existence problem of \emph{pure NEs}, where the players are restricted to pure strategies.
The algorithm described above can be adapted to that version.

\begin{theorem}
    The approximate constrained existence problem of pure NEs in stopping games is decidable in time exponential in the game size, and polynomial in the bit-size of $\epsilon$.
\end{theorem}

\begin{proof}
    The algorithm described above applies also to pure NEs.
    The only necessary modification is in the definition of the sets $\Chi_k(v)$, where the coefficients $\alpha_{ia}$ must all be equal to either $0$ or $1$.
    With that slight change, all proofs are analogous.
\end{proof}

\subsection{Hardness}

We close this subsection with a lower bound on the complexity of this problem.

\begin{restatable}[App.~\ref{app:pspace_hardness}]{theorem}{thmPSPACEHardness}\label{thm:pspace_hardness}
    The approximate constrained existence problem of NEs in stopping games is $\PSPACE$-hard, even in turn-based stopping games.
    So is the approximated constrained existence problem of pure NEs.
\end{restatable}

%% file: 4xe.tex
\subsection{Definitions and problem} 

We start this section by recalling the definition of extreme risk measures precisely, along with the definition of extreme risk-sensitive equilibria.
We assume the player set $\Pi$ is partitioned into two sets: $P$ (pessimists) and $O$ (optimists). 
    \subp{Pessimistic Risk Measure (PM)} For $i \in P$, player $i$ maximises the lowest payoff achieved with positive probability, i.e., the
 quantity    $\mathbb{PM}_i(\bsigma) = \sup \{ x \in \mathbb{R} \mid \mathbb{P}(\mu_i \ge x) = 1 \}. $
    In stopping games, this is exactly $\min \{ \mu_i(\pi) \mid \pi \text{ is a play with } \prob_\bsigma(\pi) > 0  \}$.
    
    \subp{Optimistic Risk Measure (OM)} For $i \in O$, player $i$ maximises the highest payoff achieved with positive probability:
    $$ \mathbb{OM}_i(\bsigma) = \inf \{ x \in \mathbb{R} \mid \mathbb{P}(\mu_i \le x) = 1 \} = \max \{ \mu_i(\pi) \mid \pi \text{ is a play with } \prob_\bsigma(\pi) > 0 \}.$$

We group these under the   notation $\mathbb{X}_i(\bsigma)$, referring to PM when $i \in P$ and OM if $i \in O$. 
We also write $\X(\bsigma) = (\X_i(\bsigma))_{i \in \Pi}$.

\begin{definition}[Extreme risk-sensitive equilibria (XRSE)] 
    Given a concurrent game $\mathcal{G}$ and a partition $(P, O)$, a strategy profile $\bsigma$ is an \emph{extreme risk-sensitive equilibrium (XRSE)} if  for every player $i \in \Pi$ and every alternative strategy $\sigma_i'$, we have
$\mathbb{X}_i(\bsigma_{-i}, \sigma_i') \le \mathbb{X}_i(\bsigma).$
\end{definition}

We focus on the following decision problem: 

\begin{problem}[Constrained existence problem of XRSEs]
    Given a game $\Game$, a partition $(P,O)$ of the set of players into pessimists and optimists, and two payoff profiles $\bx, \by \in \Qb^\Pi$, does there exist an XRSE $\bsigma$ in $\Game$ satisfying $x_i \leq \X_i(\bsigma) \leq y_i$ for each player $i$?
\end{problem}

We show that in concurrent stopping games, the constrained existence problem of XRSEs is in $\NP$.
The proof goes by showing that if an XRSE exists, there also exist an \emph{equivalent} XRSE, meaning that every player has the same extreme risk measure, that uses (polynomial) \emph{finite memory}.
Let us therefore define that notion.


\begin{definition}[Memory structure and memory states]
A \emph{memory structure} for a game $\mathcal{G}$ is a tuple $\mathcal{M} = (M, m_0, \alpha_{\mathsf{up}}, \bsigma^{\mathsf{choice}})$ consisting of:

\begin{itemize}
    \item a finite set $M$ of \emph{memory states}; 
    \item an \emph{initial} memory state $m_0 \in M$;
    \item a \emph{memory update function} $\alpha_{\mathsf{up}}: M \times V \times A \times V \rightarrow M$ updates the memory state based on the \emph{edge} that is taken---i.e., the source vertex, the action profile that is performed, and the new vertex reached;
    \item a tuple $\bsigma^\mathsf{choice} = (\sigma_i^\mathsf{choice})_{i\in \Pi}$ of \emph{choice functions} 
    $\sigma_i^{\mathsf{choice}}: M \times V \rightarrow \Dist(A_i)$ that output, on each vertex and in each state, the action distribution that is prescribed to player $i$.
\end{itemize}
\end{definition}

Given a memory structure $\mathcal{M}$, we extend the update function into a function $\hat{\alpha}_{\mathsf{up}}: \Hist \rightarrow M$ where $\hat{\alpha}_{\mathsf{up}}(v_0) = m_0$ and $\hat{\alpha}_{\mathsf{up}}(h \cdot \ba \cdot v) = \alpha_{\mathsf{up}}(\hat{\alpha}_{\mathsf{up}}(h), \last(h), \ba, v)$.
Then, a memory structure defines a strategy profile $\bsigma$, with $\sigma_i: h \mapsto \sigma^{\mathsf{choice}}_i(\hat{\alpha}(h))$ for each player $i$.
Note the order: when an edge is taken, the memory is updated first, and then the choice function is applied, based on the new memory state.
A strategy profile $\bsigma$ is \emph{finite-memory} if it is defined by a memory structure.





\subsection{An example}






The core of our proof is the sufficiency of polynomial finite-memory XRSEs; but the proof of that statement is significantly more involved than in the turn-based case.
Let us first use an example to give intuition on how and when memory is needed.


\begin{example} \label{ex:anchoring}
Consider the following game, with two pessimistic players $\Box$ and $\Circle$, and six vertices $\{v,v_L, v_R, t_\bot, t_\top, t_L,t_R \}$, where the terminal vertices are $t_\bot, t_\top, t_L,$ and  $t_R$, and $v$ is the starting vertex. The actions available at the start vertex are: $\{H, T, P\}$ (\emph{heads}, \emph{tails}, and \emph{punish}). 
The other vertices $v_L,v_R$ are such that each player has only one action, and therefore represented as stochastic vertices in \cref{fig:example}(a). 

In this game, is there an XRSE $\bsigma$ such that $\X_\circ(\bsigma) = \X_\square(\bsigma) = 1$?
The answer is yes, and we describe one such XRSE below. 

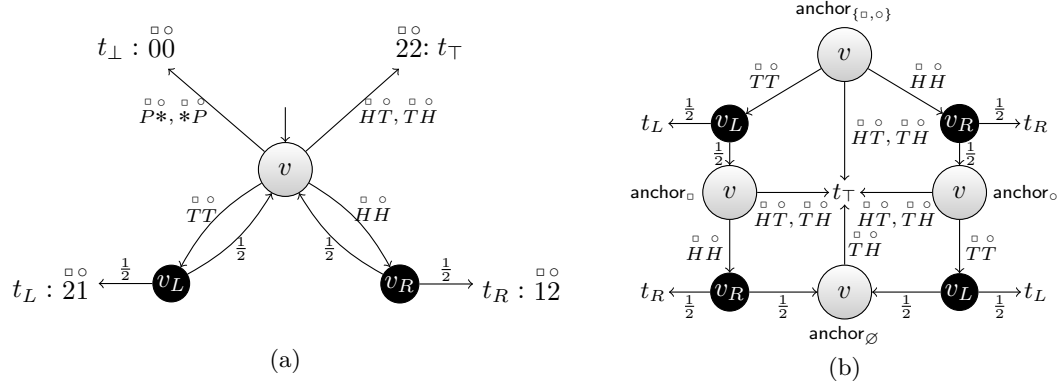
\begin{figure}
  \centering
  \resizebox{\linewidth}{!}{%
  \begin{tikzpicture}[node distance=1.5cm and 1.5cm,
      elab/.style={font=\scriptsize, inner sep=1pt},
      memlab/.style={font=\scriptsize, inner sep=1pt},
      term/.style={font=\small, inner sep=1pt}]
    \begin{scope}
      \node[circlev, initial above] (v0) {$v$};
      \node[above right=of v0] (ttop) {$\pay{2}{2} : t_\top$};
      \node[above left=of v0]  (tbot) {$t_\bot:~\pay{0}{0}$};
      \node[stoch, below right=of v0] (vR) {$v_R$};
      \node[stoch, below left=of v0]  (vL) {$v_L$};
      \node[right=0.7cm of vR] (tR) {$t_R:~\pay{1}{2}$};
      \node[left=0.7cm of vL]  (tL) {$t_L:~\pay{2}{1}$};
      \draw[->] (v0) -- (ttop) node[midway, right=1pt, elab] {$\act{H}{T},\act{T}{H}$};
      \draw[->] (v0) -- (tbot) node[midway, left=2pt, elab]  {$\act{P}{*},\act{*}{P}$};
      \draw[->] (v0) to[bend right=15] node[midway, above left=-3pt, elab]  {$\act{T}{T}$} (vL);
      \draw[->] (v0) to[bend left=15]  node[midway, above right=-3pt, elab] {$\act{H}{H}$} (vR);
      \draw[->] (vL) to[bend right=15] node[midway, below right=-3pt, elab] {$\tfrac12$} (v0);
      \draw[->] (vR) to[bend left=15]  node[midway, below left=-3pt, elab]  {$\tfrac12$} (v0);
      \draw[->] (vL) -- (tL) node[midway, above, elab] {$\tfrac12$};
      \draw[->] (vR) -- (tR) node[midway, above, elab] {$\tfrac12$};
      \node[font=\small] at (0,-2.5) {(a)};
    \end{scope}
    \begin{scope}[xshift=7.3cm, yshift=+1.5cm]
      \node[circlev] (A) at (0,0) {$v$};
      \node[memlab, above=1pt of A] {$\anc{\{\square,\circ\}}$};
      \node[term] (ttA) at (0,-1.8) {$t_\top$};
      \node[stoch] (vL1) at (-1.5,-0.9) {$v_L$};
      \node[stoch] (vR1) at ( 1.5,-0.9) {$v_R$};
      \node[term] (tL1) at (-2.5,-0.9) {$t_L$};
      \node[term] (tR1) at ( 2.5,-0.9) {$t_R$};
      \node[circlev] (B) at (-1.5,-1.8) {$v$};
      \node[memlab, left=1pt of B] {$\anc{\square}$};
        \node[circlev] (C) at ( 1.5,-1.8) {$v$};

      \node[stoch] (vL2) at (-1.5,-3.1) {$v_R$};
      \node[stoch] (vR2) at (1.5,-3.1) {$v_L$};
    \node[term] (tL2) at (-2.5,-3.1) {$t_R$};
      \node[term] (tR2) at ( 2.5,-3.1) {$t_L$};

      \node[circlev] (empv0) at (0,-3.1) {$v$};
      \node[memlab, below=1pt of empv0] {$\anc{\emptyset}$};
      \node[memlab, right=1pt of C] {$\anc{\circ}$};
      \draw[->] (A) -- (ttA) node[midway,right=1pt,elab] {$\act{H}{T},\act{T}{H}$};
      \draw[->] (A) -- (vL1) node[midway,above left=-0pt,elab]  {$\act{T}{T}$};
      \draw[->] (A) -- (vR1) node[midway,above right=-0pt,elab] {$\act{H}{H}$};
      \draw[->] (vL1) -- (tL1) node[midway,above,elab] {$\tfrac12$};
      \draw[->] (vL1) -- (B)   node[midway,left,elab]  {$\tfrac12$};
      \draw[->] (vR1) -- (tR1) node[midway,above,elab] {$\tfrac12$};
      \draw[->] (vR1) -- (C)   node[midway,right,elab] {$\tfrac12$};
      \draw[->] (B) -- (ttA) node[midway,below=1pt,elab] {$\act{H}{T},\act{T}{H}$};
      \draw[->] (C) -- (ttA) node[midway,below=1pt,elab] {$\act{H}{T},\act{T}{H}$};
      \draw[->] (B) -- (vL2) node[midway,left=1pt,elab] {$\act{H}{H}$};
      \draw[->] (C) -- (vR2) node[midway,right=1pt,elab] {$\act{T}{T}$};
    \draw[->] (vL2) -- (tL2) node[midway,below,elab] {$\tfrac12$};
      \draw[->] (vL2) -- (empv0)   node[midway,below,elab]  {$\tfrac12$};
      \draw[->] (vR2) -- (tR2) node[midway,below,elab] {$\tfrac12$};
      \draw[->] (vR2) -- (empv0)   node[midway,below,elab] {$\tfrac12$};
      \draw[->] (empv0) -- (ttA)   node[midway,right,elab, yshift=-1mm] {$\act{T}{H}$};
      \node[font=\small] at (0,-4.1) {(b)};
    \end{scope}
  \end{tikzpicture}}
  \caption{\textbf{(a)} The game of Example~\ref{ex:anchoring}. \textbf{(b)}. A representation of the strategy profile $\bsigma$, along with the memory states. Punishing plays are not depicted.}
  \label{fig:example}
\end{figure}

\subp{The strategy profile}
Since both players are pessimists, achieving $\mathbb{X}_{\circ} = \mathbb{X}_{\square} = 1$ means that each player must be given a payoff of $1$ with positive
probability. 
When the play starts, both players choose randomly between $H$ and $T$. If the vertex $v_L$ has been visited before in the history, then player $\Circle$ already had a positive probability of getting payoff $1$: then, player $\Square$ always performs $H$ and player $\Circle$ randomises between $H$ and $T$, so that player $\Square$ also gets payoff $1$ with positive probability. If $\Square$ deviates from this positional strategy, then action $P$ is played by player $\Circle$ at any further visit of $v$, to punish him.
The players behave symmetrically when $v$ is visited after $v_R$ has been seen, and when player $\Circle$ deviates.

Here, memory is used to remember who has already been given a payoff $1$ with positive probability---and who still needs to be given such a payoff.
In other words, which players are being \emph{anchored}.
We define the concept formally later; but we describe here the corresponding memory structure.

\subparagraph*{The memory structure.}
The strategy profile described above is defined by a memory structure with six states: $\{\mathsf{anchor}_{\circ,\square}$, $\mathsf{anchor}_{\circ}$, $\mathsf{anchor}_{\square}$, $
\anc{\emptyset}$, $\mathsf{punish}_{\circ}$, $\mathsf{punish}_{\square}\}$. The choice functions at $v$ are:
\begin{itemize}
  \item in state $\mathsf{anchor}_{\circ,\square}$ (initial): both players play $H$ and $T$
    uniformly at random;
  \item in state $\mathsf{anchor}_{\circ}$: player $\Circle$ plays $H$; player $\Square$
    plays $H$ and $T$ uniformly at random;
  \item in state $\mathsf{anchor}_{\square}$: player $\Square$ plays $T$; player $\Circle$ plays $H$ and $T$ uniformly at random;
  \item in state $\anc{\emptyset}$: player $\Circle$ plays $H$, player $\Square$ plays
    $T$ (the game ends in $t_\top$);
  \item in state $\mathsf{punish}_i$: the player other than $i$ plays $P$;
    player $i$'s choice is arbitrary.
\end{itemize}

The memory is updated as follows. From state $\mathsf{anchor}_{\circ,\square}$, upon
 visiting the vertex $v_R$, the memory moves to $\mathsf{anchor}_{\circ}$, and upon
visiting vertex $v_L$ to $\mathsf{anchor}_{\square}$. From state $\mathsf{anchor}_{\square}$,
upon $(H,H)\,v_R\,v$ it moves to $\anc{\emptyset}$; from
$\mathsf{anchor}_{\circ}$, upon $(T,T)\,v_L\,v$ it moves to
$\anc{\emptyset}$. If some player $i$ plays an action outside the
support prescribed by the current state, the memory moves to
$\mathsf{punish}_i$ and remains there.
 
The plays compatible with $\bar\sigma$ reach exactly the terminals
$t_\top$, $t_R$ and $t_L$, and never $t_\bot$;
hence $\mathbb{X}_{\square}(\bar\sigma) =
\mathbb{X}_{\circ}(\bar\sigma)  = 1$. Those plays are described in \cref{fig:example} (b).

\subp{This strategy profile is an XRSE}

Both players being pessimists, it suffices to 
show that whenever they deviate, they always get payoff $0$ or $1$ with positive probability;
we argue for $\Square$, the case of $\Circle$ being symmetric. 
From the second visit to vertex $v$, 
if player $\Square$ deviates, either the game has visited the vertex $v_R$, and then he anyway gets risk measure $1$; or it has visited $v_L$, and then he is prescribed to play deterministically---any deviation is detectable and he will be punished with positive probability, giving him a risk measure of $0$. 

When the game starts, deviating would mean playing $H$ or $T$ purely. 
Then, either he reaches $v_R$ first, and then receives risk measure $1$ since  $t_R$ is reached with probability $1/2$. Or, he reaches $v_L$ first, and then the memory switches to $\mathsf{anchor}_\square$, hence again $t_R$ is reached with positive probability. Either case, he has no incentive to deviate.

\subp{About the necessity of randomisation and of memory}
Critically, both randomisation and memory are required here.
For randomisation: if either player, say $\Square$, was playing $H$ (or $T$) deterministically from the beginning, then player $\Circle$ would an incentive to purely play $T$ so that the terminal $t_\top$ is reached almost surely, giving her payoff $2$.
For memory: if some player, again, say $\Square$, was always randomising between $H$ and $T$ at every occurrence of the vertex $v$, then player $\Circle$ would have a profitable deviation by playing $T$, to reach $t_\top$ or $t_R$, and get payoff $2$.
 \end{example}

\subsection{Anchored players}

We now assume given an XRSE $\bsigma$.
We aim to construct an equivalent finite-memory one.
We denote the set of histories compatible with $\bsigma$ by $H$, and the profile $\X(\bsigma)$ by $\bz$.
Given a history $h \in H$, every history $h \ba v \in H$ is called a \emph{child} of $h$.
Following from the intuition from \Cref{ex:anchoring}, we need to define formally the set of players that are \emph{anchored} after each history.
Let us first define \emph{anchorable} players.

\begin{definition}[Anchorable players]
    Let $h \in H$ be a history compatible with $\bsigma$, and let $i$ be a player.
    We say that player $i$ is \emph{anchorable} at $h$ if either:
    \begin{itemize}
        \item player $i$ is an optimist and $\X_i(\bsigma_{\restr h}) = z_i$;

        \item or player $i$ is a pessimist, and for every strategy $\sigma'_i$, we have $\X_i(\bsigma_{-i\restr h}, \sigma'_i) \leq z_i$.
    \end{itemize}
\end{definition}

However, the fact that some player $i$ is anchored at history $h$ does not say anything about how far we are from reaching the terminal where their risk measure is realised.
This is why we define the notion of \emph{$i$-rank}.

\begin{definition}[$i$-rank]
    Let $i$ be a player, and let $h$ be a history where player $i$ is anchorable.
    The history $h$ has $i$-rank $0$ if it ends in a terminal.
    Otherwise, its $i$-rank is the smallest integer $k$ such that:
    \begin{itemize}
        \item player $i$ is an optimist, and there is a child $h \ba v$ of $h$ that has $i$-rank $k-1$;

        \item or player $i$ is a pessimist, and for every action $a_i \in \Supp(\sigma_i(h))$, there is an action profile $\ba_{-i} \in \Supp(\bsigma_{-i}(h))$ and a vertex $v \in \Supp(\Delta(\last(h), \ba))$ such that the child $h \ba v$ has $i$-rank lesser than $k$.
    \end{itemize}
\end{definition}

\begin{lemma}\label{lm:irank}
    Every history where player $i$ is anchorable has an $i$-rank.
\end{lemma}

\begin{proof}
    If player $i$ is an optimist, and gets risk measure $z_i$ in the strategy profile $\bsigma_{\restr h}$, then there is an extension $\pi$ of $h$ that reaches a terminal where player $i$ gets payoff $z_i$, i.e., a history where player $i$ is anchorable and with $i$-rank $0$.
    Then the history $h$ has $i$-rank at most $|\pi| - |h|$.

    Now, let $i$ be a pessimist.
    Assume player $i$ is anchorable at some history $h$ that has no $i$-rank.
    Then, there is an action $a_i$ such that for every action profile $\ba_{-i} \in \Supp(\bsigma_{-i}(h))$, and vertex $v \in \Supp(\Delta(\last(h), \ba))$, the child $h \ba v$ has no $i$-rank.
    By iterating this reasoning from every such child, we can construct a strategy for player $i$ that guarantees that no history with an $i$-rank is ever constructed.
    Then, in particular, no history with $i$-rank $0$ is generated, i.e., no terminal where player $i$ gets payoff $z_i$ is ever reached.
    On the other hand, that strategy always performs actions in the support of the strategy $\sigma_i$, hence it never reaches a terminal where player $i$ gets less than $z_i$.
    Thus, since the game is stopping, it gives player $i$ a risk measure that is strictly greater than $z_i$, contradicting the fact that player $i$ is anchorable at $h$.
\end{proof}

We now define a labelling $\Lab: H \to 2^\Pi$; for every history $h$, we will say that the players in the set $\Lab(h)$ are \emph{anchored} at $h$.

\begin{restatable}[App~\ref{app:label}]{lemma}{lmLabel} \label{lm:label}
    There exists a labelling $\Lab: H \to 2^\Pi$ that has the following properties:
    \begin{enumerate}
        \item\label{it:labelling_root} we have $\Lab(v_0) = \Pi$;

        \item\label{it:inclusion} for every history $h$ and child $h  \ba  v$ of $h$, we have $\Lab(h  \ba v) \subseteq \Lab(h)$;

        \item\label{it:anchorable} for every history $h$, all players in $\Lab(h)$ are anchorable at $h$;

        \item\label{it:random_optimist} for every history $h$ and every optimist $i \in \Lab(h)$, there is exactly one child $h  \ba  v$ of $h$ such that $i \in \Lab(h  \ba  v)$, and $h  \ba  v$ has $i$-rank smaller than $h$;

        \item\label{it:random_pessimist} for every history $h$ and every pessimist $i \in \Lab(h)$, for every action $a_i \in \Supp(\sigma_i(h))$, there is exactly one action profile $\ba_{-i}$ and one vertex $v \in \Supp(\delta(\last(h), \ba))$ such that $i \in \Lab(h \ba v)$, and $h  \ba  v$ has $i$-rank smaller than $h$.
    \end{enumerate}
\end{restatable}

Later on, we refer to these as Properties \ref{it:labelling_root}, \ref{it:inclusion}, \ref{it:anchorable}, \ref{it:random_optimist}, and \ref{it:random_pessimist}.

\subsection{About the number of anchored sets}

We now show that the number of distinct sets that are anchored in the labelling $\Lab$ is polynomial in the size of the game.
The reader may have noted that each pessimist that randomises after a history $h$ where they are anchored are are also anchored at $|\Supp(\sigma_i(h))|$ children of $h$.
Contrary to the turn-based case, that may cause an exponential blowup: the number of anchored sets may be exponential in the number of players.
However, a main observation is that this is only possible if many players randomise simultaneously on some vertex; in which case, the space needed to represent the game is also exponential in the number of players that randomise.

\begin{restatable}{lemma}{lmCounting}
\label{lm:counting}
    Let $p$ be the number of players in $\Game$.
    Then, the number of anchored sets is smaller than or equal to $(p^2 + 1) \| \Game \|$.
\end{restatable}

\begin{proof}
    Given two sets $A$ and $B$, we say that $A$ and $B$ are \emph{comparable} if we have either $A \subseteq B$ or $B \subseteq A$, and that they \emph{cross} if they are not comparable but satisfy $A \cap B \neq \emptyset$.
    For each vertex $v$, let $C_v$ be the \emph{core} of $v$, i.e., the set of pessimists $i$ that have at least two distinct actions available in $v$.
    We now define the \emph{game core} as the set of subsets of $\Pi$ that are included in some vertex core, i.e., the set $\Core = \bigcup_{v \in V} 2^{C_v}$.
    Our proof relies on the following argument: the size of the game core is bounded by the game size.

    \begin{proposition} \label{prop:size_core}
        We have $|\Core| \leq \| \Game \|$.
    \end{proposition}

    \begin{proof}
        We have
        $|\Core| \leq \sum_v \left| 2^{C_v} \right|= \sum_v 2^{|C_v|}.$
        For each $v$, there are at least $2^{|C_v|}$ action profiles possible, hence the space required to represent the transition function is at least this sum.
    \end{proof}

    We now count separately the anchored sets in the game core, and those outside.

    \subp{In the game core}
    The number of sets $A \in \Core$ that are anchored is bounded by $|\Core|$.
    By \Cref{prop:size_core}, that quantity is smaller than or equal to the game size $\|\Game\|$.

    \subp{Outside of the game core}
    Let us first notice the following.

    \begin{proposition}\label{prop:crossing_sets}
        Let $A$ and $B$ be two crossing anchored sets.
        Then, we have $A \cap B \in \Core$.
    \end{proposition}

    \begin{proof}
        Let $h_A, h_B \in H$ be two histories such that $\Lab(h_A) = A$ and $\Lab(h_B) = B$.
        None of those histories is a prefix of the other one: otherwise, by Property~\ref{it:inclusion}, we would have $\Lab(h_A) \subseteq \Lab(h_B)$ or $\Lab(h_B) \subseteq \Lab(h_A)$, i.e. the sets $A$ and $B$ would be comparable, and thus would not cross.
        
        Let now $h$ be the longest common prefix of $h_A$ and $h_B$: the history $h$ has then two distinct children $h \ba v$ and $h \ba' v'$ such that $h \ba v$ is a prefix of $h_A$ and $h \ba' v'$ is a prefix of $h_B$.
        By Property~\ref{it:inclusion}, we have $A \subseteq \Lab(h \ba v)$ and $B \subseteq \Lab(h \ba' v')$, hence $A \cap B \subseteq \Lab(h \ba v) \cap \Lab(h \ba' v')$.

        Let then $i \in A \cap B$.
        By Properties~\ref{it:random_optimist}, player $i$ is necessarily a pessimist, and by Property~\ref{it:random_pessimist} we necessarily have $a_i \neq a_i'$.
        Thus, player $i$ belongs to the core of the vertex $\last(h)$.
        Which proves the inclusion $A \cap B \subseteq C_{\last(h)}$.
    \end{proof}

    Let us now call \emph{critical set} a set $C \subseteq \Pi$ such that $C \not\in \Core$, but $C \setminus \{i\} \in \Core$ for every player $i$.
    We then get the following result.

    \begin{proposition}
        Let $A$ and $B$ be two anchored sets.
        If there is a critical set $C$ with $C \subseteq A \cap B$, then the sets $A$ and $B$ are comparable.
    \end{proposition}

    \begin{proof}
        Since we have $\emptyset \in \Core$, the set $C$ is nonempty.
        The sets $A$ and $B$ are therefore either crossing or comparable.
        But if they were crossing, we would have $A \cap B \in \Core$, and therefore $C \in \Core$, which is false.
        Therefore, they are comparable.
    \end{proof}

    Consequently, for every critical set $C$, there can be at most $p$ distinct anchored sets containing $C$.
    On the other hand, there are at most $|\Core|p$ critical sets (to choose a critical set, one must choose a core set $C'$ and then a player $i$ such that $C' \cup \{i\} \not\in \Core$).
    Consequently, there are at most $|\Core|p^2 \leq \|\Game\|p^2$ anchored sets outside the core.

    \subp{Conclusion}
    There are at most $\|\Game\| + \|\Game\|p^2 = (p^2+1)\|\Game\|$ anchored sets, as desired.
\end{proof}

\subsection{A finite-memory strategy profile}

The labelling $\Lab$ provides us with the fundamental structure from which we can construct a finite-memory strategy profile.
\paragraph*{Tool box: memoryless strategies} Before we proceed to the lemma, we recall one lemma 
from a recent work~\cite{BHMST26}, that shows that if one player deviates in a way that can be observed by the other players, then the other players can form a team and punish the deviator using a memoryless randomised strategy. Earlier works~\cite{dAH00,AHK02} had considered concurrent games where players can form coalitions, but assumed that the 
players in a coalition are allowed to correlate their random choices, which makes it possible to reduce them to one meta-player;
that is not the case in this setting.

 \begin{restatable}[App.~\ref{app:value}]{lemma}{ValuePunish}\label{lem:value}\label{lem:punishment}
For every player $i$ and vertex $v$, the value at vertex $v$, written
 $\val_v(\bsigma) = \inf_{\substack{\bsigma_{-i}}}\ \sup_{\sigma_i}\  \X_i\!\big(\bsigma_{-i},\sigma_i\big)
$, where $\sigma_i$ and $\bsigma_{-i}$ range over strategies from vertex $v$, 
 is realised by a memoryless strategy profile.
 \end{restatable}

We also 
note a well-known result from the result from MDP literature that states that 
against a memoryless strategy profile, positional strategies are optimal.

\begin{lemma}[Positional strategy for one player]\label{lem:oneplayerMDP}
Let $\bsigma_{-i}$ be a memoryless strategy profile.
 Then, the supremum  $\sup_{\sigma_i} \X_i(\bsigma_{-i},\sigma_i)$ is realised by a positional strategy.
 \end{lemma}

\paragraph*{The construction}

\begin{restatable}[App.~\ref{app:finite_memory}]{lemma}{lmFiniteMemory} \label{lm:finite_memory}
    There exists a finite-memory strategy profile $\bsigma^\star$ equivalent to $\bsigma$, with at most $p + (p^2+1) \|\Game\|$ memory states.
\end{restatable}

\begin{proof}[Proof sketch]
    We construct from $\bsigma$ and $\Lab$ a memory structure with the states $\punish_i$, for each player $i$ ,
        and $\anchor_{A}$, for each anchored set $A$.
By \Cref{lm:counting}, we have at most $p + (p^2+1) \|\Game\|$ states.
    The initial state is $\anchor_\Pi$.

    
    Each memory state corresponds to a memoryless strategy until the memory is updated.
     In the state $\mathsf{punish}_i$, all players follow the strategy profile $\bsigma^{i}$, where, for each player $i$, we let $\bsigma_{-i}^i$ be the memoryless punishing strategy  from \Cref{lem:punishment} and where $\sigma_i^i$ is  an arbitrary memoryless strategy.
   
At memory state $\anc{A}$, we define the choice function. For every history $h \in H$, we define its \emph{absolute rank} $R(h)$ as the maximum $i$-rank of $h$, for $i$ ranging over $\Lab(h)$,
with the convention $R(h) = 0$ when $\Lambda(h) = \emptyset$. For memory state $\anc{A}$ and vertex $v$ we pick a history $h_A^v \in H$ with last vertex $v$, of minimal absolute rank. At state $\anchor_A$, from $v$ 
each player $i$ outputs the distribution $\sigma_i(h_A^v)$. 

Finally, let us define the update function.
From the state $\punish_i$, the memory always remains in the state $\punish_i$. From the state $\anchor_A$, when seeing an edge $(v, \ba, w)$, if there is a player $i$ such that $a_i \not\in \Supp(\sigma_i(h_A^v))$, then the memory switches to $\punish_i$. Otherwise, it switches to $\anchor_{\Lab(h_A^v \ba w)}$.

    

    We show in the appendix that this strategy profile is an XRSE and is equivalent to $\bsigma$.
\end{proof}






\subsection{Conclusion}
\begin{restatable}{theorem}{Npcomplete}
    \label{thm:Npcomplete}
    The constrained existence problem for XRSE in a concurrent stopping game is $\NP$-complete. 
\end{restatable}

\begin{proof}
    Hardness is already known~\cite{BHT25}.
    As for easiness, by \Cref{lm:finite_memory}, we know that there exists an XRSE satisfying the specified constraints if and only if there is a finite-memory one, using at most $p + (p^2+1) \|\Game\|$ memory states.
    Such a strategy profile $\bsigma^\star$ can be guessed in polynomial time: what remains to be proven is that we can check in polynomial time that the generated risk measures satisfy the constraints, and that it is an XRSE.

    The strategy profile $\bsigma^\star$ induces a Markov chain.
    The profile $\bz = \X(\bsigma)$ can be computed by computing the set of terminals that are reached with positive probability, by standard reachability algorithms.
    Then, for each player $i$, the strategy profile $\bsigma^\star_{-i}$ induces an MDP.
    Deciding whether player $i$ has a profitable deviation in $\bsigma^\star$ amounts to deciding whether they can ensure reaching the set $T^\star = \{t \in T \mid \mu_i(t) > z_i\}$ with positive probability (if player $i$ is an optimist) or almost surely (if player $i$ is a pessimist).
    Both questions can be answered with standard polynomial-time algorithms~\cite[Section~10.6.1]{BK08}.
\end{proof}

%% file: 5Disc.tex
We have seen two different routes to circumvent the undecidability of the constrained existence of Nash equilibria in  concurrent  stopping games~\cite{UW11}.
The first
keeps the solution concept fixed and relaxes exactness: our algorithm decides the approximate problem in time exponential in the game and polynomial in the bit-size of~$\varepsilon$, with no restrictions on the memory of the players, and is accompanied by a $\PSPACE$ lower bound already in turn-based games. The second keeps exactness and
changes the players' risk attitude, replacing expected payoff by an extreme
measure; for this notion the constrained existence problem remains $\NP$-complete, even under concurrency. 


In both cases, a first limitation is that our algorithms are valid only in stopping games.
The approximate algorithm, in particular, relies heavily on that hypothesis, and decidability of the approximate problem in general concurrent games remains open.
A second limitation is the complexity gap the we leave for that approximate problem, which is now known to be $\PSPACE$-hard and in the class $\EXPTIME$: a polynomial-space algorithm, if it exists, would have to resort to completely different techniques.

Finally, we believe that a central question for future works is the exploration of equilibrium notions for which the constrained existence problem is decidable, with no restrictions on the strategies.
So far, XRSEs are the only such notion, but constitute a strong idealisation of the players' preferences, as they do not consider the probability distribution of payoffs at all---only its support.
Variants of XRSEs could be considered: for example, their \emph{subgame-perfect} version, where players maximise their risk measure after every possible history and not only in the whole game.
But the literature also suggests less \emph{extreme} risk measures, such as the value at risk (VaR), which can be seen as a generalisation of the pessimistic risk measure, where the players have some \emph{tolerance parameter} $p$ and can ignore bad payoffs occurring with probability lesser than $p$.
Decidability of the constrained existence problem of equilibria defined from that notion would then open new possibilities for the algorithmic study of equilibria in stochastic games.

%% file: 6Appendix3.tex
\subsection{Proof of \Cref{lm:truncation}} \label{app:truncation}

\lmTruncation*

\begin{proof}
    Let us build the strategy profile $\bsigma'$ as follows.
    On every history $h$ of length at most $kN$, we define $\bsigma'(h) = \bsigma(h)$.
    Let now $h$ be a history of length $kN$. For every history $h'$ of which the history $h$ is a prefix, we define $\bsigma'(h')$ arbitrarily.
    
    We know that there exists a memoryless strategy profile $\btau^h$ such that for every terminal vertex $t$, we have $\prob^{\last(h)}_{\btau}(\Diamond t) = \prob_{\bsigma}(\Diamond t \mid h)$~\cite[Theorem~3.2]{EKVY08}. But obtaining such exact strategies might not possible without correlated strategies.

    By construction, the strategy profile $\bsigma'$ has memory horizon $kN$, and each player $i$ has the same expected payoff as in $\bsigma$.
    Let us prove that $\bsigma'$ is a $\delta$-NE.

    Let $i$ be a player and let $\sigma''_i$ be a strategy for player $i$.
    We want:
    $$\Eb_i(\bsigma'_{-i}, \sigma''_i) \leq \Eb_i(\bsigma') + \delta.$$
    Let us decompose:
    $$\Eb_i(\bsigma'_{-i}, \sigma''_i) = \prob_{\bsigma'_{-i}, \sigma''_i}(V^{\leq kN} T) \Eb_i(\bsigma'_{-i}, \sigma''_i \mid V^{\leq kN} T) + \prob_{\bsigma'_{-i}, \sigma''_i}(V^{kN+1}) \Eb_i(\bsigma'_{-i}, \sigma''_i \mid V^{kN+1}).$$
    Note that on every history of length at most $kN$, the strategy profile $\bsigma'$ behaves exactly as $\bsigma$.
    Therefore, the probabilities of reaching or avoiding terminal vertices within the first $kN+1$ steps are the same in the two strategy profiles.
    Similarly, assuming a terminal vertex is reached within $kN+1$ steps, the expected payoff for player $i$ is the same.
    On the other hand, every expected payoff in the game $\Game$ is smaller than or equal to $R$.
    Hence:
    \begin{equation}\label{eq:expected_payoff_dev}
        \Eb_i(\bsigma'_{-i}, \sigma''_i) \leq \prob_{\bsigma_{-i}, \sigma''_i}(V^{\leq kN} T) \Eb_i(\bsigma_{-i}, \sigma''_i \mid V^{\leq kN} T) + \frac{\delta}{2R} R.
    \end{equation}
    Let us now apply the same decomposition when following the strategy profile $(\bsigma_{-i}, \sigma''_i)$:
    $$\Eb_i(\bsigma_{-i}, \sigma''_i) = \prob_{\bsigma_{-i}, \sigma''_i}(V^{\leq kN} T) \Eb_i(\bsigma_{-i}, \sigma''_i \mid V^{\leq kN} T) + \prob_{\bsigma'_{-i}, \sigma''_i}(V^{kN+1}) \Eb_i(\bsigma_{-i}, \sigma''_i \mid V^{kN+1}).$$
    And since every expected payoff is greater than or equal to $-R$:
    $$\Eb_i(\bsigma_{-i}, \sigma''_i) \geq \prob_{\bsigma_{-i}, \sigma''_i}(V^{\leq kN} T) \Eb_i(\bsigma_{-i}, \sigma''_i \mid V^{\leq kN} T) - \frac{\delta}{2R}R.$$
    By injecting this in \Cref{eq:expected_payoff_dev}, we obtain:
    $$\Eb_i(\bsigma'_{-i}, \sigma''_i) \leq \Eb_i(\bsigma_{-i}, \sigma''_i) + \frac{\delta}{2} + \frac{\delta}{2} = \Eb_i(\bsigma_{-i}, \sigma''_i) + \delta.$$
    Finally, since the strategy profile $\bsigma$ is an NE, we obtain:
    $$\Eb_i(\bsigma'_{-i}, \sigma''_i) \leq \Eb_i(\bsigma) + \delta = \Eb_i(\bsigma') + \delta.$$
    The strategy profile $\bsigma'$ is an $\delta$-NE.
\end{proof}

\subsection{Proof of \Cref{lm:Chin_meaning}} \label{app:Chin_meaning}

\lmChinMeaning*

\begin{proof}
    We proceed by induction on $k$.

    \subparagraph*{The case $k=0$.}
    Consider first the case $k=0$: strategy profiles with memory horizon $0$ are exactly the memoryless strategy profiles, hence if $\bsigma$ is a memoryless strategy profile from $v$, the cube $C \in \Cubes$ containing the vector $\bchi^\bsigma$ is included in the set $\Chi_0(v)$.
    Conversely, every vector $\bchi \in \Chi_0(v)$ is contained in a cube $C \in \Cubes$ that also contains a characteristic vector $\bchi^\bsigma$ where $\bsigma$ is a memoryless strategy profile: then, we have $\| \bchi - \bchi^\bsigma \| \leq \frac{R}{D}$.
    Let us now prove that the statement holds for $k>0$.

    \subparagraph*{For $k>0$, the set $\Chi_k(v)$ contains the characteristic vectors.}
    Let $\bsigma$ be a strategy profile from $v$, with memory horizon $k$.
    For every action profile $\ba$ available from $v$, and for every $w \in V$, the strategy profile $\btau^{\ba w}: (w, \ba_1, v_1 \dots, \ba_k, v_k) \mapsto \bsigma(v, \ba, w, \dots, v_k)$ has memory horizon $k-1$, and by induction hypothesis we have $\chi^{\ba w} = \chi^{\btau^{\ba w}} \in \Chi_{k-1}(w)$.
    Then, for every $i$, we have:
    $$\chi^\bsigma_{i \reg} = \Eb_i(\bsigma) = \sum_{\ba \in \Av(v)} \sum_{w \in V} \left(\prod_{j \in \Pi} \alpha_{aj} \right) \Delta(v, \ba)(w) \Eb_i(\btau^{\ba w}),$$
    where $\alpha_{aj} = \sigma_j(v)(a_j)$.
    An analogous equality holds for $\chi^\bsigma_{i \dev}$, hence the cube containing the vector $\bchi^\bsigma$ is contained in the set $\Chi_k(v)$, and therefore we have $\chi^\bsigma \in \Chi_k(v)$.

    \subparagraph*{For $k>0$, every vector in $\Chi_k(v)$ is close to some characteristic vector.}
    Conversely, let $\bchi \in \Chi_k(v)$, and let $C \in \Cubes$ be the cube containing $\bchi$.
    Then, the cube $C$ also contains a vector $\bchi'$ for which vectors $\bchi^{\ba w}$ and coefficients $\alpha_{ia}$ satisfying the equalities in \Cref{def:Chin}.
    By induction hypothesis, for each pair $(\ba, w)$, there is a strategy profile $\btau^{\ba w}$ with memory horizon $k-1$ satisfying:
    $$\left\| \bchi^{\ba w} - \bchi^{\btau^{\ba w}} \right\| \leq k\frac{R}{D}.$$
    Consider then the strategy profile $\bsigma$, defined by $\sigma_i(a) = \alpha_{ia}$ for each player $i$ and each action $a$, and that follows the strategy profile $\bsigma^{\ba w}$ after having performed the action profile $\ba$ and seen the vertex $w$.
    That strategy profile has memory horizon $k$.
    Moreover, for every $i$, we have:
    $$\left|\chi'_{i\reg} - \chi^\bsigma_{i\reg}\right| = \left|\sum_{\ba \in \Av(v)} \sum_{w \in V} \Delta(v, \ba)(w) \left(\prod_j \alpha_{ja_j} \right) \left( \chi^{\ba w}_{i\reg} -  \chi^{\btau^{\ba w}}_{i\reg} \right)\right|$$
    $$\leq \sum_{\ba \in \Av(v)} \sum_{w \in V} \Delta(v, \ba)(w) \left(\prod_j \alpha_{ja_j} \right) \left| \chi^{\ba w}_{i\reg} -  \chi^{\btau^{\ba w}}_{i\reg} \right|$$
    $$\leq \sum_{\ba \in \Av(v)} \sum_{w \in V} \Delta(v, \ba)(w) \left(\prod_j \alpha_{ja_j} \right) k\frac{R}{D} = k\frac{R}{D}.$$
    Analogously, we get $\left|\chi'_{i\dev} - \chi^\bsigma_{i\dev}\right| \leq k\frac{R}{D}$, and therefore $\left\| \bchi' - \bchi^\bsigma\right\| \leq k\frac{R}{D}$.
    Finally, using the triangular inequality, we have:
    $$\left\| \bchi - \bchi^\bsigma\right\| \leq \left\| \bchi - \bchi'\right\| + \left\| \bchi' - \bchi^\bsigma\right\|$$
    $$\leq \frac{R}{D} + k \frac{R}{D} = (k+1) \frac{R}{D},$$
    as desired.
\end{proof}

\subsection{Proof of \Cref{lm:computing_Chin}} \label{app:computing_Chin}

\lmComputingChin*

\begin{proof}
    Each set $\Chi_\l(v)$ is a union of cubes in $\Cubes$, and can therefore be described by the list of those cubes, which takes space $O(D^{2|\Pi|})$, which grows exponentially with the instance size (note that $D$ is already exponential in the instance size---but that does not affect the result).
    Assuming $\Chi_{\l-1}$ is known, computing the set $\Chi_\l(v)$ can be done by iterating all cubes $C \in \Cubes$ and all cubes $C_{\ba w} \in \Chi_{\l-1}(w)$ for each pair $(\ba, w)$, and then by checking whether there exist vectors $\bchi \in C$ and $\bchi^{\ba w} \in C_{\ba w}$ satisfying the following sentence, written in the existential theory of the reals:
    $$\left( \sum_{a  \in \Av_i(v)} \alpha_{ia} = 1 \right)$$
    $$\wedge \bigwedge_{i \in \Pi} \left( \chi_{i\reg} = \sum_{\ba \in \Av(v)} \sum_{w \in V} \Delta(v, \ba)(w) \left(\prod_j \alpha_{ja_j} \right) \chi^{\ba w}_{i\reg} \right)$$
    $$\wedge \bigwedge_{i \in \Pi} \bigwedge_{a_i \in \Av_i(v)} \left( \chi_{i\dev} \geq \sum_{\ba_{-i} \in \Av_{-i}(v)} \sum_{w \in V} \Delta(v, \ba)(w) \left(\prod_{j \neq i} \alpha_{ja_j} \right)  \chi^{\ba w}_{i\dev} \right)$$
    $$\wedge \bigwedge_{i \in \Pi} \bigvee_{a_i \in \Av_i(v)} \left( \chi_{i\dev} = \sum_{\ba_{-i} \in \Av_{-i}(v)} \sum_{w \in V} \Delta(v, \ba)(w) \left(\prod_{j \neq i} \alpha_{ja_j} \right)  \chi^{\ba w}_{i\dev} \right).$$

    This sentence has polynomial size (let us remember that each set $\Av_i(v)$ is explicitly described in the instance, even though it may be exponential in the number of players), and the cubes $C$ and $C_{\ba w}$ are also described by a polynomial number of inequations.
    Since the existential theory of the reals is decidable in $\PSPACE$ \cite{DBLP:conf/stoc/Canny88}, only polynomial space, and therefore exponential time, is needed to decide the existence of such a solution.
    The set $\Chi_\l(v)$ can therefore be constructed from $\Chi_{\l-1}$ in exponential time.
    By iterating that algorithm from $\l=1$ to $\l=k$, we obtain an exponential-time algorithm computing the set $\Chi_k(v)$.
\end{proof}

\subsection{Proof of \Cref{thm:pspace_hardness}} \label{app:pspace_hardness}

\thmPSPACEHardness*

    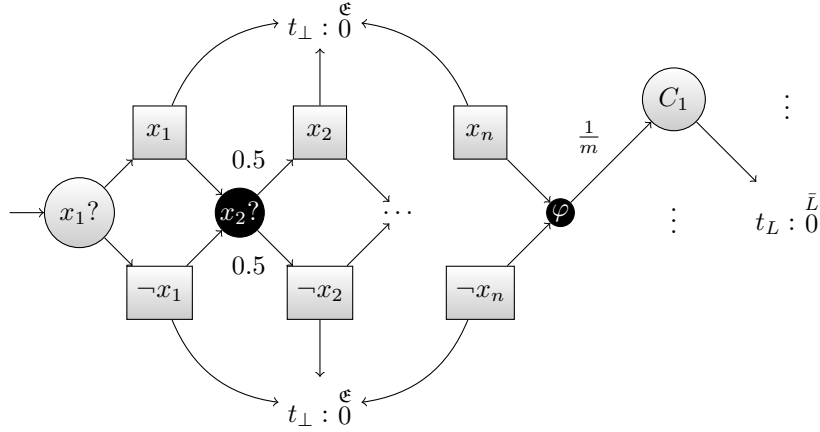
\begin{figure}
    \centering
    \begin{tikzpicture}[node distance = 2cm]
        \node[circlev, initial left] (x1?) {$x_1?$};
        \node[squarev, above right of=x1?] (x1) {$x_1$};
        \node[squarev, below right of=x1?] (nx1) {$\neg x_1$};
        \node[stoch, below right of=x1] (x2?) {$x_2?$};
        \node[squarev, above right of=x2?] (x2) {$x_2$};
        \node[squarev, below right of=x2?] (nx2) {$\neg x_2$};
        \node[below right of=x2] (dots) {\dots};
        \node[squarev, above right of=dots] (xn) {$x_n$};
        \node[squarev, below right of=dots] (nxn) {$\neg x_n$};
        \node[stoch, below right of=xn] (phi) {$\phi$};
        \node[right of=phi] (dots2) {\vdots};
        \node[circlev, above of=dots2] (C1) {$C_1$};
        \node[right of=C1] (dots3) {\vdots};
        \node[below of=dots3] (tL) {$t_L:~\stackrel{\bar{L}}{0}$};
        \node[above of=x2] (tb1) {$t_\bot:~\stackrel{\Eve}{0}$};
        \node[below of=nx2] (tb2) {$t_\bot:~\stackrel{\Eve}{0}$};

        \path[->] (x1?) edge (x1);
        \path[->] (x1?) edge (nx1);
        \path[->] (x1) edge (x2?);
        \path[->] (nx1) edge (x2?);
        \path[->] (x2?) edge node[above left] {$0.5$} (x2);
        \path[->] (x2?) edge node[below left] {$0.5$} (nx2);
        \path[->] (x2) edge (dots);
        \path[->] (nx2) edge (dots);
        \path[->] (x1) edge[bend left] (tb1);
        \path[->] (x2) edge (tb1);
        \path[->] (xn) edge[bend right] (tb1);
        \path[->] (nx1) edge[bend right] (tb2);
        \path[->] (nx2) edge (tb2);
        \path[->] (nxn) edge[bend left] (tb2);
        \path[->] (xn) edge (phi);
        \path[->] (nxn) edge (phi);
        \path[->] (phi) edge node[above left] {$\frac{1}{m}$} (C1);
        \path[->] (C1) edge (tL);
    \end{tikzpicture}
    \caption{A reduction from QBF}
    \label{fig:qbf}
\end{figure}

\begin{proof}
    We proceed by reduction from the $\PSPACE$-complete QSAT.
    Let us assume a formula $\phi = \exists x_1 \forall x_2 \dots \exists x_{n-1} \forall x_n \bigwedge_{i=1}^m C_i$, where each clause $C_i$ is a conjunction of three literals over the variables $x_1, \dots, x_n$.
    We construct a game $\Game$ and an error term $\epsilon$ such that if the formula $\phi$ is true, then there is a (pure) NE where the player Eve has expected payoff at least 1, and that if the formula $\phi$ is false, then there is no $\epsilon$-NE where Eve has expected payoff at least $1-\epsilon$.

    \paragraph*{The construction}
    The construction is depicted by \Cref{fig:qbf}.
    The game $\Game$ has $2n+1$ players, called $x_1$, $\neg x_1$, \dots, $x_n$, $\neg x_n$, and Eve, denoted $\Eve$.
    Intuitively, Eve tries to build a valuation satisfying all clauses; later on, for each clause, she is asked to select a literal that is satisfied by the generated valuation.
    Each literal player, on the other hand, will have a profitable deviation if and only if Eve wrongly claims that its negation is satisfied.
    On \Cref{fig:qbf}, all omitted rewards are $1$, and the terminal vertex $t_\bot$ has been depicted twice for clarity.
    Round vertices are controlled by Eve, square vertices by literal players, and black vertices are \emph{stochastic}: each player has only one action available.
    
    For each variable $x_k$, we define three vertices denoted by $x_k?$, $x_k$, and $\neg x_k$.
    If the variable $x_k$ is quantified existentially, then the vertex $x_k?$ is controlled by Eve, and she chooses deterministically to move the game to $x_k$ or to $\neg x_k$.
    If it is quantified universally, then the vertex $x_k?$ is stochastic: the game goes to $x_k$ or to $\neg x_k$ with probability $\frac{1}{2}$ each.

    Each literal vertex $L \in \{x_k, \neg x_k\}$ is controlled by the player $L$.
    There, that player can either go to the vertex $t_\bot$, or (if $k < n$) to the vertex $x_{k+1}$, or (if $k=n$) to the vertex $\phi$.

    The vertex $\phi$ is stochastic: from there, the game goes with uniform probability distribution to the vertices $C_1, \dots, C_m$.
    Each vertex $C_i$ is controlled by Eve, and from the vertex $C_i$, Eve can reach the terminal vertices $t_L$ for all literals $L$ of $C_i$.

    In the terminal vertex $t_\bot$, every player gets payoff $1$, except Eve, who gets payoff $0$.
    In the terminal vertex $t_L$, every player gets payoff $1$, except player $\bar{L}$ (the negation of $L$), who gets payoff $0$.
    The initial vertex is $x_1?$.
    We define $\epsilon = \frac{1}{3 \times2^n m + 1}$.

    \paragraph*{If the formula $\phi$ is true, then there is a (pure) NE where Eve gets expected payoff $1$.}

    Let us assume that the formula $\phi$ is true: then, there exist mappings $f_k: \{0, 1\}^{k-1} \to \{0,1\}$, for each odd $k$, such that for every valuation $\nu$ of the variables $x_2, x_3, \dots, x_n$, the valuation $\nu'$ that extends $\nu$ with $\nu'(x_k) = f_k(x_1, \dots, x_{k-1})$ for each even $k$ satisfies the conjunction $\bigwedge_i C_i$.

    Let $\bsigma$ be the strategy profile where the literal players never go to the vertex $t_\bot$, and where Eve plays according to the mappings $f_k$, i.e., after the history $hx_k?$, goes deterministically to the vertex $x_k$ if $f_k(b_1, \dots, b_{k-1}) = 1$ and to $\neg x_k$ otherwise, where $b_j = 1$ if the history $h$ visits the vertex $x_j$ and $0$ otherwise.
    Then, when reaching a vertex $C_i$ after a history $h$, that history defines a valuation of all variables, and by definition of the mappings $f_k$, that valuation satisfies $C_i$: Eve goes then deterministically to a terminal vertex $t_L$ such that the literal $L$ is satisfied.

    In this strategy profile, Eve has expected payoff $1$, sine the vertex $t_\bot$ is never reached.
    Therefore, she has no profitable deviation.
    As for every literal player $L$, every play compatible with $\bsigma$ that doe snot give player $L$ the payoff $1$, i.e. that reaches the terminal vertex $t_{\bar{L}}$, defines a valuation that satisfies the literal $\bar{L}$.
    Therefore, the vertex $L$ is not visited, and player $L$ has no profitable deviation.
    That strategy profile is a pure NE.

    \paragraph*{If the formula $\phi$ is false, then there is no $\epsilon$-NE where Eve gets expected payoff at least $1-\epsilon$.}

    Let $\bsigma$ be a strategy profile where Eve gets expected payoff at least $1-\epsilon$, and let us prove that it is not an $\epsilon$-NE.
    In this strategy profile, the probability of reaching $t_\bot$ is less than $\epsilon$.

    Let us define the mappings $f_k: \{0, 1\}^{k-1} \to \{0,1\}$, for each odd $k$, by $f_k(b_1, \dots, b_{k-1}) = 1$ if and only if after the history that visits exactly the vertices $x_j$ such that $b_j = 1$, Eve goes to the vertex $x_k$ with probability at least $\frac{1}{2}$.
    Since the formula $\phi$ is false, there is at least one valuation $\nu$ such that the corresponding valuation $\nu'$ does not satisfy the conjunction $\bigwedge_i C_i$, i.e., does not satisfy one specific clause $C_i$.
    Let $h$ be the history corresponding to that valuation, and that ends in the clause vertex $C_i$: under the strategy profile $\bsigma$, that history is followed with probability at least $\frac{(1-\epsilon)^n}{2^nm}$.
    Then, let $t_L$ be a terminal vertex that is reached with probability at least $\frac{1}{3}$.
    Along that play, player $\bar{L}$ gets payoff $0$.
    Since the valuation $\nu$ does not satisfy the clause $C_i$, it does not satisfy the literal $L$, which means that the vertex $\bar{L}$ has been visited along the history $h$.
    Then, by going to the vertex $t_\bot$, player $\bar{L}$ has a deviation that is profitable by at least:
    $$\frac{(1-\epsilon)^n}{3 \times 2^n m} \leq \frac{1-\epsilon}{3 \times 2^n m}$$
    $$= \frac{1-\frac{1}{3 \times2^n m + 1}}{3 \times 2^n m} = \frac{1}{3 \times2^n m + 1} = \epsilon.$$
    Therefore, the strategy profile $\bsigma$ is not an $\epsilon$-NE.
\end{proof}

%% file: 7Appendix4.tex
\subsection{Proof of \cref{lem:value}} \label{app:value}

\ValuePunish*

   \begin{proof}Let: $$z = \inf_{\substack{\bsigma_{-i} \\ \text{from } v}}\ \sup_{\sigma_i}\  \X_i\!\big(\bsigma_{-i},\sigma_i\big)$$

The infimum is realised, since extreme risk measure range over a finite set of values.
Therefore, there exists a strategy $\bsigma_{-i}$ such that for every strategy $\sigma_i$, we have $\X_i\left(\bsigma_{-i},\sigma_i\right) \leq z$.

If player $i$ is an optimist, and since the game is stopping, this means that the coalition $\Pi \setminus \{i\}$ has a strategy profile that guarantees almost-sure reachability of the set $\{t\in T:\mu_i(t)\le z\}$ (where all players randomise independently).
Then, by~\cite[Theorem~5]{BHMST26}, there exists a memoryless strategy profile with that property.

Similarly, if player $i$ is a pessimist, the coalition $\Pi \setminus \{i\}$ has a strategy profile that guarantees that the set $\{t\in T:\mu_i(t)\le z\}$ is reached with positive probability.
Then, by~\cite[Theorem~1]{BHMST26}, there exists a memoryless strategy profile with that property.
\end{proof}

\subsection{Proof of \Cref{lm:label}} \label{app:label}

\lmLabel*

\begin{proof}
We construct the labelling $\Lambda$ from the strategy $\bsigma$, and prove by an induction on the length of the histories that this function satisfies the properties required.

\subparagraph{Base case.}
Consider the history $v_0$. We define $\Lambda(v_0) = \Pi$.
This trivially satisfies Property~\ref{it:labelling_root}. For Property~\ref{it:anchorable}, observe that the definition of anchorability ensures that optimists get their optimistic expectation $z_i$, while for pessimists, it requires that they have no profitable deviation. Therefore Property~\ref{it:anchorable} follows since $\bsigma$ is an XRSE. We show that the labelling satisfies the rest of the properties in the next step of the proof.
\subparagraph{Induction step: $\Lambda(h)$ is defined for a history $h$}
Suppose the labelling $\Lab$ has been defined until a history $h$, then we define it for all children $h\ba w$. 
Since $\Lambda(h)$ is defined until history $h$ and satisfies Properties 1-5, we know that all players in $\Lambda(h)$ are anchorable. Since player $i \in \Lambda(h)$ is anchorable, we know from \cref{lm:irank} that the history $h$ has an $i$-rank. Recall that by the definition of $i$-rank, if player $i$ is an optimist, there is a child that has $i$-rank that is $k-1$ and if player $i$ is a pessimist, for every action $a_i\in \Supp(\sigma_i(h))$, there is an action profile $\ba_{-i}\in \Supp(\bsigma_{-i}(h))$ and a vertex $v\in \Supp(\Delta(\last(h),\ba))$ such that $h \ba v$ has $i$-rank less than $k$.

We first start by declaring $\Lambda(h\ba w) = \emptyset$ for each $\ba$ and $w$, and add a player to each extension one after another for each player as below.

\noindent \textbf{The player $i$ is an optimist.} 
We then know from \cref{lm:irank} there is a child that has $i$-rank that is $k-1$. We pick such an extension and add player $i$ to $\Lab(h \ba w)$.

\noindent\textbf{The player $i$ is a pessimist.}
For every action $a_i\in \Supp(\sigma_i(h))$, there is an action profile $\ba_{-i}\in \Supp(\bsigma_{-i}(h))$ and a vertex $v\in \Supp(\Delta(\last(h),\ba))$ such that $h \ba v$ has $i$-rank less than $k$.
We add player $i$ to the set $\Lab(h\ba w)$ for exactly one extension $h \ba v$  for each action $a_i$.


We now show that this label $\Lambda$ satisfies the Properties 2-5.
\begin{itemize}
    \item[2.]From construction we already know that Property \ref{it:inclusion} is satisfied. 
    \item[3.]We prove Property~\ref{it:anchorable} by first showing that every $i\in \Lambda(h\ba w)$ is anchorable. Any $i\in \Lambda(h\ba w)$ was added because the history $h$ had a finite $i$-rank, and therefore by definition of $i$-rank, it is anchorable.
\item[4.] Property~\ref{it:random_optimist} follows from construction of how optimists were added to $\Lambda(h\ba w)$.
\item[5.] Similarly, Property~\ref{it:random_pessimist} follows from construction of how pessimists were added to $\Lambda(h\ba w)$.
\end{itemize}

\end{proof}

\subsection{Proof of \Cref{lm:finite_memory}} \label{app:finite_memory}
\lmFiniteMemory*

\begin{proof}
We first define the strategy profile, then prove it is equivalent to $\bsigma$, and finally that it is an XRSE.

\paragraph*{Construction of the strategy profile $\bsigma^\star$}

We construct from $\bsigma$ and $\Lab$ a memory structure with the following states:
    \begin{itemize}
        \item for each player $i$, the state $\punish_i$;

        \item for each anchored set $A$, the state $\anchor_{A}$.
    \end{itemize}

    The initial state is $\anchor_\Pi$.
    By \Cref{lm:counting}, we have at most $p + (p^2+1) \|\Game\|$ states.
    
    

    Let us now define the choice functions.
    For each player $i$, let $\bsigma_{-i}^i$ be the memoryless punishing strategy from \Cref{lem:punishment}.
    Let $\sigma_i^i$ be an arbitrary memoryless strategy.
    In the state $\mathsf{punish}_i$, all players follow the strategy profile $\bsigma^{i}$.

    For every history $h \in H$, we define its \emph{absolute rank} as
$$
  R(h) \;=\; \max_{j \in \Lambda(h)} j\text{-rank of history }h,
$$
with the convention $R(h) = 0$ when $\Lambda(h) = \emptyset$. By
Property~\ref{it:anchorable}, every player $j \in \Lambda(h)$ is anchorable at $h$, and
therefore has a finite $j$-rank by Lemma~\ref{lm:irank}; hence
$R(h) \in \mathbb{N}$. For every anchored set $A$ and every vertex $v$
such that the set
$H_{A,v} \;=\; \{\, h \in H \mid \Lambda(h) = A
    \text{ and } \last(h) = v \,\}$
is nonempty, we fix a representative
$
  h^v_A \;\in\; \argmin_{h \in H_{A,v}} R(h),
$
which exists since $\mathbb{N}$ is well-ordered, and we write
$\Phi(A, v) = R(h^v_A)$. Representatives are needed only for the pairs
$(A, v)$ that are reachable under the update function defined below;
the claim in the proof of
Proposition~\ref{prop:pess} shows that $H_{A,v} \neq \emptyset$
holds for all such pairs. On the remaining pairs, the choice functions
are defined arbitrarily. Then we define the choice function for $\anchor_A$: when seeing a vertex $v$, each player $i$ outputs the distribution $\sigma_i(h_A^v)$.
    
    Finally, let us define the update function.
    From the state $\punish_i$, the memory always remains in the state $\punish_i$.
    From the state $\anchor_A$, when seeing an edge $(v, \ba, w)$, if there is a player $j$ such that $a_j \not\in \Supp(\sigma_j(h_A^v))$, then the memory switches to $\punish_j$.
    Otherwise, it switches to $\anchor_{\Lab(h_A^v \ba w)}$.

\paragraph*{The strategy profile $\bsigma^\star$ is equivalent to $\bsigma$}
\begin{proposition}\label{prop:preserve}
    We have $\X_i(\bsigma^\star) = z_i$.
\end{proposition}

Let us recall that we had $\bz = \X(\bsigma)$.
Let $i$ be a player, and let us prove that $\X(\bsigma^\star) = z_i$.

Whenever the strategy profile is in a memory state $\anchor_A$ with $i \in A$ and see the vertex $v$, the players play the distribution $\bsigma(h_A^v)$.
In a play where there is no deviation, the punish state is not entered. By construction, all terminals that are reachable are reached also by the original strategy $\bsigma$. For an optimist $i$, this implies that $\X_i(\bsigma^\star) \leq z_i$ and for pessimists $\X_i(\bsigma) \geq z_i$.

The other direction of the proof is structured around these claims.
 
\begin{claim}[Rank reduces]\label{cl:descent}
  Let $h \in H$, let $h\bar{a}w$ be a child of $h$, and let
  $j \in \Lambda(h\bar{a}w)$. Then
  $\irank{j}{h\bar{a}w} < \irank{j}{h}$.
\end{claim}
\begin{claimproof}
  By Property~2, we have $j \in \Lambda(h)$. If $j$ is an optimist,
  then by Property~\ref{it:random_optimist} the history $h\bar{a}w$ is the unique child of $h$
  whose label contains $j$, and that child has $j$-rank smaller than
  $h$. If $j$ is a pessimist, then $a_j \in \Supp(\sigma_j(h))$, and by
  Property~\ref{it:random_pessimist}, applied to the action $a_j$, the history $h\bar{a}w$ is
  the unique child of $h$ with $j$th action $a_j$ whose label contains
  $j$; that child has $j$-rank smaller than $h$. 
\end{claimproof}

\begin{claim}[Payoffs at anchored terminals]\label{cl:terminal}
  Let $c \in H$ be a history ending in a terminal vertex $t$, and let
  $i \in \Lambda(c)$. Then $\mu_i(t) = z_i$.
\end{claim}
\begin{claimproof}
    This claim follows from the definition of anchored sets and ranks. 
\end{claimproof}

In both cases of optimists and pessimists, we show that there is a positive probability of obtaining $z_i$. 
Let $i$ be a player. We construct, step by step, a play prefix
compatible with $\bar{\sigma}^\star$ and of positive probability,
visiting pairs $(A_t, v_t)$ with $i \in A_t$ and
$H_{A_t, v_t} \neq \emptyset$, where the memory state at vertex $v_t$
is $\anc{A_t}$. We initialise with $(A_0, v_0) = (\Pi, v_0)$,
which satisfies those conditions by Property~1. (Assume $v_0 \notin T$.)
 
Assume the pair $(A_t, v_t)$ has been constructed with
$v_t \notin T$, and write $h_t = h^{v_t}_{A_t}$. If $i$ is an
optimist, let $c_t = h_t\bar{a}w$ be the unique child of $h_t$ whose
label contains $i$, given by Property~\ref{it:random_optimist}. If $i$ is a pessimist, pick
an action $a_i \in \Supp(\sigma_i(h_t))$, and let $c_t = h_t\bar{a}w$
be the unique child of $h_t$ with $i$th action $a_i$ whose label
contains $i$, given by Property~\ref{it:random_pessimist}. In both cases,
$\bar{a} \in \Supp(\bar{\sigma}(h_t))$, which is exactly the support
prescribed by $\bar{\sigma}^\star$ in the memory state
$\anc{A_t}$ at $v_t$, and $w \in \Supp(\Delta(v_t, \bar{a}))$:
the edge $(v_t, \bar{a}, w)$ is taken with positive probability, no
player leaves the prescribed support, and the memory is updated to
$\anc{\Lambda(c_t)}$, with $i \in \Lambda(c_t)$.

The play stops at $t\in T$.
The prefix built so far, a finite concatenation
of positive-probability transitions, has positive probability, and by
Claim~\ref{cl:terminal} for $c_t$, player $i$ receives the
payoff $z_i$. Otherwise, we set $(A_{t+1}, v_{t+1}) = (\Lambda(c_t),
w)$: the history $c_t$ witnesses $H_{A_{t+1}, v_{t+1}} \neq
\emptyset$, and by minimality of the representative $h_{t+1}$ and
Claim~\ref{cl:descent}, we have $R(h_{t+1}) \le R(c_t) < R(h_t)$. The
sequence $(R(h_t))_t$ is a strictly decreasing sequence of natural
numbers: the construction therefore stops, after at most $R(h_0)$
steps, at a terminal vertex where player $i$ receives the payoff
$z_i$.



\paragraph*{The strategy profile $\bsigma^\star$ is an XRSE.}

We split this part into two propositions and prove them in 
\cref{prop:opt,prop:pess}.
We now state a claim that we use in \cref{prop:opt,prop:pess}. Throught the claim and the propositions, we write $\rho^A$ to be the memoryless strategy played at memory state $\anchor_A$ by $\bsigma^\star$. 
\begin{restatable}[Post-deviation values are capped]{claim}{cappingvalues}\label{cl:devvalue}
Let $h$ be a $\bsigma^{\star}$-compatible history whose memory state is
$\mathsf{anchor}_A$ at $v=\last(h)$, with $i\in A$. We have 
\begin{enumerate}
\item[(i)] \textup{(optimist)} if player $i\in O$, then $\mathrm{val}_i(w)\le z_i$
for every joint action $\ba$ with $\ba_{-i}\in\Supp(\rho^{A}_{-i}(v))$ and
every $w\in\Supp(\Delta(v,\ba))$;
\item[(ii)] \textup{(pessimist)} if $i\in P$, then for every action $a_i$
available to player $i$ at $v$ there exist a joint action $\ba$ with
$\ba_i=a_i$ and $\ba_{-i}\in\Supp(
\rho^{A}_{-i}(v))$ and a successor
$w\in\Supp(\Delta(v,\ba))$ such that $\mathrm{val}_i(w)\le z_i$.
\end{enumerate}
    
\end{restatable}
\begin{claimproof}
Suppose $h$ be a $\bsigma^{\star}$-compatible history whose memory state is
$\mathsf{anchor}_A$ at $v=\last(h)$, with $i\in A$. Consider the candidate history $h_A^v$ whose last vertex is $v$, that was chosen to define the strategy at vertex $v$.
Such a $h_A^v$ is an anchorable history compatible with $\bsigma$ and $\Lambda(h_A^v) = A$. We call this history $h' = h_A^v$. 

We note the following statement that we use: if some $w\in\Supp(\Delta(v,\ba))$ with
$\ba_{-i}\in\Supp(\bsigma_{-i}(h'))=\Supp(\rho^{A}_{-i}(v))$ has
$\mathrm{val}_i(w)>z_i$, then since $\mathrm{val}_i(w)$ is the least value
the opponents can enforce on $i$ from $w$ and $\bsigma_{-i\restr h'\ba w}$
is one such opponent profile,
$$\sup_{\tau_i}\X_i(\bsigma_{-i\restr h'\ba w},\tau_i)\ge\mathrm{val}_i(w)>z_i,$$
so player $i$ has a strategy $\theta_i^{\ba w}$ with
$\X_i(\bsigma_{-i\restr h'\ba w},\theta_i^{\ba w})>z_i$.

\emph{(i).} \textbf{Optimists. } Consider some vertex $w$ in the game as in the statement has $\mathrm{val}_i(w)>z_i$,
and let $\theta_i$ be the corresponding strategy above. Consider the
deviation $\tau_i$ that agrees with $\bsigma_i$ along $h'$, plays $\ba_i$
at $v$, and follows $\theta_i$ once $(\ba,w)$ is realised. Under
$(\bsigma_{-i},\tau_i)$ the history $h'$ is reached with positive
probability; at $v$ the opponents play $\bsigma_{-i}(h')$, whose support is
$\Supp(\rho^{A}_{-i}(v))$, so the pair $(\ba_{-i},w)$ is realised with
positive probability; and from $w$ player $i$ secures a payoff $>z_i$ with
positive probability. As $\X_i$ is the greatest payoff obtained with
positive probability for optimists, $\X_i(\bsigma_{-i},\tau_i)>z_i$, contradicting that
$\bsigma$ is an XRSE for optimists. Hence $\mathrm{val}_i(w)\le z_i$.

\emph{(ii).} \textbf{Pessimists. } Fix an available action $a_i$ and suppose, for contradiction,
that $\mathrm{val}_i(w)>z_i$ for every realised pair $(\ba,w)$ with
$\ba=(a_i,\ba_{-i})$, $\ba_{-i}\in\Supp(\bsigma_{-i}(h'))$ and
$w\in\Supp(\Delta(v,\ba))$. For each such pair pick $\theta_i^{\ba w}$ as
above, and let $\tau_i$ play $a_i$ at $v$ and then follow $\theta_i^{\ba w}$
from $(\ba,w)$. The pessimistic risk
is the minimum of the realisable payoffs, so
$$
\X_i(\bsigma_{-i\restr h'},\tau_i)
=\min_{(\ba,w)}\X_i(\bsigma_{-i\restr h'\ba w},\theta_i^{\ba w})>z_i,
$$
contradicting that $i$ is anchored at $h'$. Hence some
vertex $w$ that is reached by $\bsigma$ has $\mathrm{val}_i(w)\le z_i$.
\end{claimproof}

\begin{restatable}{proposition}{optimistdeviation}\label{prop:opt}
No optimist has a profitable deviation in $\bsigma^{\star}$.
\end{restatable}
\begin{proof}
Let player $i\in O$ be an optimist and let $\sigma_i'$ be a deviation, with
$z'=\X_i(\bsigma^{\star}_{-i},\sigma_i')$. Along any play compatible with
$\bsigma^{\star}_{-i}$ the shared memory can only
\begin{itemize}
\item stay among states $\mathsf{anchor}_A$ for a fixed $A$;
\item move from $\mathsf{anchor}_A$ to $\mathsf{anchor}_B$ with
$B\subsetneq A$;
\item move to $\mathsf{punish}_i$, and then remain there.
\end{itemize}
Since the anchored sets strictly shrink along the second kind of step and
$\Pi$ is finite, and since $\Game$ is stopping, every play stabilises
either in $\mathsf{punish}_i$ or in some $\mathsf{anchor}_A$. As $\X_i$
equals the greatest payoff that player $i$ obtains with positive probability, it suffices to bound by $z_i$ the payoff of every play that
occurs with positive probability under $(\bsigma^{\star}_{-i},\sigma_i')$.

If such a play never enters $\mathsf{punish}_i$, then player $i$ never
played outside its prescribed support, so the play is compatible with
$\bsigma^{\star}$ itself; since we have shown  \cref{prop:preserve}, every payoff that
$\bsigma^{\star}$ realises with positive probability is at most $z_i$ for
the optimist~$i$.

If the play enters $\mathsf{punish}_i$, let $v$ be the last vertex at which
the memory is still an anchoring state $\mathsf{anchor}_A$, and let
$v\,\ba\,w$ be the step that triggers the switch: it means $i$ played the
out-of-support action $\ba_i$ while the opponents played within
$\Supp(\bsigma_{-i}(h_A^v))$. From $w$ onward the opponents follow the
value-optimal punishment $\bsigma^{\dagger i}_{-i}$, so every continuation
gives player $i$ at most
$\sup_{\tau_i}\X_i\big((\bsigma^{\dagger i}_{-i},\tau_i)_{\restr w}\big)\le\mathrm{val}_i(w)$,
which is at most $z_i$ by \cref{cl:devvalue}(i). Hence this play, too,
yields player $i$ at most $z_i$.

In both cases $z'\le z_i$, so $\sigma_i'$ is not a profitable deviation.
\end{proof}

\begin{restatable}{proposition}{pessimistdeviation}\label{prop:pess}
No pessimist has a profitable deviation in $\bsigma^{\star}$.
\end{restatable}
\begin{proof}
Let $i\in P$ and let $\sigma_i'$ be a deviation; by \cref{lem:oneplayerMDP}
we may take $\sigma_i'$ positional on the support arena. We show that
player $i$ still obtains a payoff $\le z_i$ with positive probability, i.e.\
$\PM_i(\bsigma^{\star}_{-i},\sigma_i')\le z_i$, so the deviation does not
raise $i$'s pessimistic value above $z_i$. By construction, and as in \cref{prop:opt}, along any play compatible with $\bsigma^{\star}_{-i}$ the shared memory stabilises in $\mathsf{punish}_i$ or in some $\mathsf{anchor}_A$.

\emph{A detectable deviation.} Suppose that at
some history $h$ compatible with $\bsigma^{\star}$, whose memory state is
$\mathsf{anchor}_A$ with $i\in A$ at $v=\last(h)$, the strategy $\sigma_i'$
plays an action $a_i\notin\Supp(\rho^{A}_i(v))$, so the memory switches to
$\mathsf{punish}_i$. By \cref{cl:devvalue}(ii) some successor $w$ reached by
this action satisfies $\mathrm{val}_i(w)\le z_i$; this $w$ occurs with
positive probability, and from $w$ the punishment $\bsigma^{\dagger i}_{-i}$
guarantees, against every strategy of player~$i$, a positive-probability
play of payoff $\le z_i$ (\cref{lem:value}, case~(P)). Hence
$\X_i(\bsigma^{\star}_{-i},\sigma_i')\le z_i$, and the deviation is not
profitable.

\emph{No detectable deviation.} Otherwise, the strategy $\sigma'_i$ never leaves the prescribed
support at an anchoring state carrying $i$: whenever the memory is in
a state $\anc{A}$ with $i \in A$ at a vertex $v$, every action
played with positive probability by $\sigma'_i$ lies in
$\Supp(\sigma_i(h^v_A))$. The argument of \cref{prop:preserve} is used on the profile
$(\bar{\sigma}^\star_{-i}, \sigma'_i)$: at each such pair $(\anc{A}, v)$, pick an action $a_i$ played with positive
probability by $\sigma'_i$; since $a_i \in \Supp(\sigma_i(h^v_A))$,
Property~\ref{it:random_pessimist} ensures that the unique child $h^v_A\bar{a}w$ with $i^\text{th}$ action
$a_i$ whose label contains $i$, and its opponent part $\bar{a}_{-i}$
lies in $\Supp(\bar{\sigma}_{-i}(h^v_A))$, which is exactly the support played
by $\bar{\sigma}^\star_{-i}$. The corresponding edge is therefore
taken with positive probability under
$(\bar{\sigma}^\star_{-i}, \sigma'_i)$, and the rank of the
representatives strictly decreases as in the proof of
Proposition~\ref{prop:preserve}. Hence, with positive probability, a
terminal history $c$ with $i \in \Lambda(c)$ is reached, at which
player $i$ receives the payoff $z_i$ by
Claim~\ref{cl:terminal}. Consequently,
$X_i(\bar{\sigma}^\star_{-i}, \sigma'_i) \le z_i$, and the deviation
is not profitable.

\end{proof}
\end{proof}